\documentclass[12pt]{amsart}
\usepackage{amssymb,slashed, color,array,verbatim}
\usepackage{amsmath,amsfonts,amsthm,euscript,mathrsfs,multirow}
\usepackage[all]{xy}
\usepackage[english]{babel}
\usepackage{epsf}
\usepackage{graphicx}
\csname @addtoreset\endcsname{equation}{section}
\newcommand{\SU}{\mathop{\rm SU}}

\newcommand{\Sl}{\mathop{\rm {}SL} }

\newtheorem{theorem}{Theorem}[section]
\newtheorem{lemma}[theorem]{Lemma}
\newtheorem{corol}[theorem]{Corollary}
\newtheorem{prop}[theorem]{Proposition}

\theoremstyle{definition}
\newtheorem{definition}[theorem]{Definition}
\theoremstyle{remark}
\newtheorem{remark}[theorem]{Remark}

\newenvironment{notation and conventions}{\textbf{Notation and conventions.}}{ }

\DeclareFontFamily{U}{rsf}{} \DeclareFontShape{U}{rsf}{m}{n}{ <5> <6> rsfs5 <7> <8> <9> rsfs7 <10-> rsfs10}{}
\DeclareMathAlphabet\Scr{U}{rsf}{m}{n}
\definecolor{pink}{rgb}{1,0,1}

\DeclareSymbolFont{extraup}{U}{zavm}{m}{n}
\DeclareMathSymbol{\varheartsuit}{\mathalpha}{extraup}{86}
\DeclareMathSymbol{\vardiamondsuit}{\mathalpha}{extraup}{87}

\begin{document}

\begin{titlepage}
\begin{center}
\baselineskip=10pt{\LARGE
 $D_5$ elliptic fibrations: non-Kodaira fibers and new orientifold limits of F-theory\\
}
\vspace{2 cm}
{\large  Mboyo Esole$^{\spadesuit,
 \varheartsuit}$ , James Fullwood$^\clubsuit$ and
 Shing-Tung  Yau$^{\spadesuit}$
  } \\
\vspace{.6 cm}
${}^\spadesuit$Department of Mathematics, \ Harvard University, Cambridge, MA 02138, U.S.A.\\
${}^\varheartsuit$Jefferson Physical Laboratory, Harvard University, Cambridge, MA 02138, U.S.A.\\
${}^\clubsuit$Mathematics Department, Florida State University, Tallahassee, FL 32306, U.S.A.\\

\end{center}

\vspace{1cm}
\begin{center}

{\bf Absract}
\vspace{.3 cm}
\end{center}

{\small

A $D_5$ elliptic fibration is a fibration whose generic fiber is modeled by the complete intersection of two quadric surfaces in $\mathbb{P}^3$. They  provide simple examples of elliptic fibrations  admitting a rich spectrum of singular fibers (not all on the list of Kodaira) without introducing singularities in the total space of the fibration and therefore avoiding a discussion of their resolutions. We study systematically the fiber geometry of such fibrations using Segre symbols and compute several topological invariants. 

We present for the first time  Sen's (orientifold) limits for $D_5$ elliptic  fibrations. These orientifolds limit  describe different weak coupling limits of F-theory to type IIB string theory giving  a system of three brane-image-brane pairs in presence of a $\mathbb{Z}_2$ orientifold. The orientifold theory is mathematically described by the double cover the  base of the elliptic fibration. 
Such orientifold theories are characterized by a transition from a semi-stable singular fiber to an unstable one. In this paper, we describe the first example of a weak coupling limit in F-theory characterized by a transition to a non-Kodaira (and non-ADE) fiber. 
Inspired by string dualities, we obtain non-trivial topological relations connecting the elliptic fibration and the different loci that appear in its weak coupling limit. Mathematically, these are very surprising relations relating the total Chern class of the $D_5$ elliptic fibration and those of different loci that naturally appear in the weak coupling limit. We work in arbitrary dimension and are result don't assume  the Calabi-Yau condition.

\vfill

$^\spadesuit$Email:\   {\tt    esole at  math.harvard.edu
,  yau  at math.harvard.edu}
\par
$^\clubsuit$Email:\   {\tt jfullwoo at math.fsu.edu}
}

\end{titlepage}
\addtocounter{page}{1}
 \tableofcontents{}
\newpage

\section{Introduction and summary}

\subsection{F-theory and type IIB string theory}
 Calabi-Yau varieties were first introduced in compactification of string theory to  geometrically engineer ${\mathcal  N}=1$ supersymmetry in four spacetime dimension\cite{CHSW, SW}.  The best understood configurations are perturbative in nature and have a constant value of the axio-dilaton field. 
The axio-dilaton field is  a  complex scalar particle $\tau=C_0+ i e^{-\phi}$ ($i^2=-1$) where the axion $C_0$ is the Ramond-Ramond zero-form of the $D(-1)$-brane (the D-instanton) while the dilaton $\phi$ determines the string coupling $g_s$ via its exponential $g_s=e^{\phi}$. Due to the positivity of the string coupling, the axio-dilaton resides exclusively in the complex upper half-plane.  In type IIB string theory, the S-duality group is the modular group $\Sl(2,\mathbb{Z})$ under which the axio-dilaton field $\tau$ transforms  as the period modulus of an elliptic curve  $\mathbb{C}/ (\mathbb{Z}+\tau\mathbb{Z})$.
{
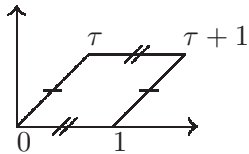
\begin{figure}[thb]
\begin{center}
\setlength{\unitlength}{1.6 mm}
\begin{picture}(15,9)(0,0)
\put(0,0){
\qbezier(0,0)(4,0)(15,0)
\qbezier(0,0)(0,4)(0,10)
\qbezier(0,0)(3,3)(6,6)
\qbezier(6,6)(10,6)(14,6)
\qbezier(8,0)(11,3)(14,6)
\put(15,0){\qbezier(0,0)(-.5,.2)(-.7,.7)
\qbezier(0,0)(-.5,-.2)(-.7,-.7)
}
\multiput(3,3)(8,0){2}{\qbezier(-.7,0)(0,0)(.7,0)}
\multiput(3.7,0)(6,6){2}{\qbezier(-.7,-.7)(0,0)(.7,.7)}
\multiput(4.3,0)(6,6){2}{\qbezier(-.7,-.7)(0,0)(.7,.7)}
\put(5,6.8){ \text{\small $\tau$}}
\put(13,6.8){ \text{\small $\tau+1$}}
\put(8,-2){\text{\small $1$}}
\put(0,-2){\text{\small $0$}}
\put(0,10){\qbezier(0,0)(.2,-.5)(.7,-.7)\qbezier(0,0)(-.2,-.5)(-.7,-.7)}
}
\end{picture}
\caption{A torus seen as the quotient $\mathbb{C}/ (\mathbb{Z}+\tau \mathbb{Z})$.}
\end{center}
\end{figure}
}

F-theory \cite{Vafa:1996xn,Morrison:1996na,Morrison:1996pp,Bershadsky:1996nh} provides a description of compactifications of type IIB string theory on  non-Calabi-Yau varieties $B$ endowed with  a varying axio-dilaton field. 
The power of F-theory is that it elegantly encapsules  non-perturbative aspects of  type IIB string theory compactified on a space $B$ using the mathematics of elliptic fibrations over $B$ to describe the variation of the axio-dilaton field and the action of S-duality. 
As such, type IIB string theory compactified on a space  $B$ with a varying axio-dilaton  is geometrically engineered in  F-theory  by an elliptically fibered space  $\varphi:Y\rightarrow B$. 
When the  base of the fibration is  of
complex dimension $d$, it corresponds to a compactification  to a $(10-2d)$ real dimensional space-time $\mathscr{M}_{10-2d}$. 
The most common cases studied in the literature are compactifications  to six and four spacetime dimensions and they are described respectively in F-theory by elliptic threefolds and fourfolds. 

The non-vanishing first Chern class of the compact space $B$ is balanced by the presence of $(p,q)$ 7-branes\footnote{ A $(p,q)$ 7-brane is a dynamical brane extended in seven space dimensions and characterized by the fact that  $(p,q)$ strings (bounds states of $p$ fundamental strings and $q$ D1 branes with $p$ and $q$ relatively prime integer numbers) can end on it. A $(1,0)$ 7-brane is the usual D7-brane of perturbative string theory while the other $(p,q)$ 7-branes are non-perturbative solitonic branes that can  be obtained from a D7-brane  by the action of S-duality, which in type IIB is the modular group $\Sl(2,\mathbb{Z})$.} wrapping non-trivial divisors of $B$ so that  supersymmetry is preserved after the compactification. The  presence of $(p,q)$ 7-branes induces    non-trivial $\Sl(2,\mathbb{Z})$ monodromies of the axio-dilaton field for which 7-branes are magnetic sources. 
Although the compactification space $B$ seen by type IIB  is not  Calabi-Yau, the total space $Y$ of the elliptic fibration $\varphi:Y\rightarrow B$ is required to be  Calabi-Yau \cite{Vafa:1996xn}. This is most naturally seen using the M-theory picture of F-theory\footnote{ M-theory  compactified on an elliptic fibration $\varphi:Y\rightarrow B$ to a spacetime $\mathscr{M}_{9-2d}$ is dual to type IIB compactified on the base $B$ of the elliptic fibration to a spacetime $S^1 \times\mathscr{M}_{9-2d}$ with non-trivial three-form field strength on $B\times S^1$. The radius of the circle $S^1$ being inversely proportional to the area of the elliptic fiber.  As  we take the limit of zero area, we end up with type IIB string theory on  $B\times \mathscr{M}_{10-2d}$.  }. From the type IIB perspective, we would also like the fibration to  admit a section $s:B\rightarrow Y$ so that the compactification  space $B$  is unambiguously  identified within the elliptic fibration itself:
$$
\xymatrix{
T^2\ \ar[r] &Y\ar[d]^{\varphi}\\
& B \ar@/^/[u]^{s  }
}
$$
The existence of a section is not necessary from the point of view of M-theory. For a review of F-theory, we refer to \cite{Denef:2008wq}. 
The singular fibers of the elliptic fibration play an essential role in the dictionary between physics and geometry \cite{Vafa:2009se}. For example, one can use elliptic fibrations to geometrically engineer sophisticated gauge theories with matter representations and Yukawa couplings all specified by the geometry of the elliptic fibration \cite{Vafa:1996xn, Morrison:1996na, Bershadsky:1996nh, GrassiMorrison2}.  

\subsubsection*{F-theory and the mathematics of elliptic fibrations}

From a mathematical point of view, F-theory provides a fresh perspective on the geometry of elliptic fibrations with a rich inflow of new problems, conjectures and perspective inspired by  physics. However, these questions can be attacked purely mathematically and open new ways to think of elliptic fibrations connecting it to representation theory and other area of mathematics. For example, the duality between F-theory and the Heterotic string has motivated the study of principle holomorphic $G$-bundles over elliptic fibrations by Freedmann-Morgan-Witten \cite{FMW1, FMW2}. Since the work of Kodaira on elliptic surfaces \cite{Kodaira}, it is well appreciated that ADE-like Dynkin diagrams appear as singular fibers of an elliptic fibration over codimension-one loci in the base. F-theory associates to these ADE diagrams  specific gauge theories living on branes wrapped around the location of the singular fibers in the base \cite{Vafa:1996xn,Morrison:1996na,Bershadsky:1996nh}. Non-simply-laced Lie groups also appear naturally once we consider the role of monodromies and distinguish between split and non-split singular fibers \cite{Bershadsky:1996nh}. When the base of the fibration is higher dimensional, matter representations are naturally associated with certain loci in codimension-two in the base over which singular fibers enhance \cite{GrassiMorrison2}. 
An analysis on the condition for anomaly cancellations of the gauge theories described in F-theory leads to surprising relations involving representations of the gauge group and the topology of the Chow ring of the elliptic fibration \cite{GrassiMorrison2}. 

The description of phenomenological models in F-theory, like for example the $\SU(5)$ Grand Unified model \cite{Vafa:2009se}, has motivated the study of elliptic fibrations that admit a discriminant locus with wild singularities and a rich structure of singular fibers that enhance to each other as we consider higher codimension loci in the discriminant locus \cite{EY}. Such enhancements  often violate  the conditions that will typically be required by mathematicians studying elliptic fibrations \cite{Nakayama.Global, Miranda, Szydlo}. For  example, the discriminant locus of the elliptic fibration is usually not supposed to be a divisor with normal crossing \cite{EY} and non-flat fibrations can lead to very interesting physics like for example the presence of massless stringy objects \cite{Codim3}.  
With the appearence of non-Kodaira fibers in elliptic fibration over a higher dimensional base \cite{EY}, the dictionary between singular fibers and physics have to made more precised \cite{MorrisonTaylor,Marsano:2011hv, GrassiMorrison2}. Under the more general conditions considered in physics, there is not yet a classification of the possible singular fibers of a higher dimensional  elliptic fibration. See  \cite{EY}, for more information.

The connection between F-theory and its type IIB weak coupling limit  uncovers interesting geometric relations involving the elliptic fibrations and a double cover of its base.  Sen has shown that the weak coupling limit of F-theory can be an orientifold theory\cite{Sen.Orientifold}. Sen's construction  is mathematically described by  certain degenerations of the elliptic fibration organized by  transitions from semi-stable to unstable singular fibers \cite{AE2}.  The presence of charged objects in a compact space leads to cancellation relations in physics  known as tadpole conditions. These tadpole conditions are sophisticated version of the familiar Gauss theorem in electromagnetism that ensures that the total charge in a compact space is zero. Using dualities between F-theory and type IIB string theory,  tadpole relations will induce  non-trivial relations  between the topological invariants of different  varieties that appear in the description of Sen's weak coupling limit \cite{CDE,AE1,AE2}. This has motivated the introduction of a new Euler characteristic inspired by string dualities to deal with some of the singularities\cite{CDE,AE1,AE2}.

\subsection{Synopsys}

We would like to explore the physics of the weak coupling limit of F-theory in presence of non-Kodaira fibers. One easy way to do it without dealing with resolution of singularities to generate non-Kodaira fibers is to start with certain families of elliptic fibrations that naturally admit such fibers. In this way we will be able to provide the first example of a weak coupling limit of F-theory involving non-Kodaira fibers. 

\subsubsection*{$D_5$ elliptic fibrations}
In this article, we  continue the work started in \cite{AE2} and explore different aspects of elliptic fibrations beyond the realm of Weierstrass models. Non-Weierstrass models provide new ways of describing the strong coupling limit of certain non-trivial type IIB orientifold compactifications with brane-image-brane pairs by embedding them in F-theory. We consider elliptic fibrations whose generic fiber is an elliptic curve modeled by the complete intersection of two quadric surfaces in $\mathbb{P}^3$. Such fibrations are referred to in the physics literature as {\em $D_5$ elliptic fibrations} \cite{KLRY,Andreas:1999ty,KMT,CHL}. An equivalent description of the  generic fiber of a $D_5$ elliptic fibration is to see it as the  base locus of  a {\em pencil of quadrics} in $\mathbb{P}^3$. This little change of perspective provides powerful tools to describe the singular fibers of $D_5$ elliptic fibrations, since pencils of quadrics are  naturally classified by {\em  Segre symbols} as we will review in section \ref{Section.Segre}.  We therefore classify  all the singular fibers of a smooth $D_5$ elliptic fibration using Segre symbols. Singular fibers of an elliptic surface were classified by Kodaira and are described by Kodaira symbols.
In the context of $D_5$ elliptic fibrations, Segre symbols provide a finer description of the singular fibers than symbols of Kodaira since they detect the  degree of each of the components of a given singular fiber. In the study of $D_5$ elliptic fibrations, the geometric objects at play  are very classical: quadric surfaces, conics, twisted cubics and elliptic curves. It follows that the study of $D_5$ elliptic fibrations is reduced to a promenade in the garden of 19th century Italian school of algebraic geometry where we can pick up all the ingredients we need.

For $D_5$ elliptic fibrations, non-Kodaira singular fibers appear innocently  without introducing singularities in the total space and so avoiding any resolutions of singularities. We will explore the physical relevance of these non-Kodaira singular fibers from the point of view of the weak coupling limit of F-theory\cite{Sen.Orientifold, AE2}. We will analyze some degenerations of these fibrations and deduce non-trivial topological relations. The degeneration we obtain describes a theory of an orientifold  with three brane-image-brane pairs, two of which are in the same homology class as the orientifold. The cancellation of the D3 tadpole provides a non-trivial relations between Euler characteristic of the elliptic fibration and the Euler characteristic of the divisors corresponding to the orientifold and the brane-image-brane pairs. We will prove that the same relation holds at the level of the total Chern class of these loci.  We will see that the non-Kodaira fibers indicate a certain regime in which the orientifold and the two brane-image-brane pairs that are in its  homology class coincide.

One might think that F-theory leads only to mathematical results for  Calabi-Yau elliptic  fourfolds and threefolds since these are the usually varieties in which F-theory is relevant physically. Strikingly enough, many of the insights gained on the structure of elliptic fibrations coming from F-theory are true  without any assumptions on the dimension of the base and without assuming the Calabi-Yau condition \cite{FMW1, AE1, AE2, EY, GrassiMorrison2}, providing yet another example of why string theory is a source of inspiration for geometers. Therefore, although our considerations are  inspired by F-theory,    we will not restrict ourself to Calabi-Yau elliptic fourfolds or threefolds but will work with elliptic fibrations over a base of arbitrary dimension and without assuming the Calabi-Yau condition. We will provide rigourous mathematical  proof to all the  geometric and  topological statements inspired from F-theory.

\subsubsection*{An historical note on pencils of quadrics and Segre symbols}
The classification of pencils of quadrics is indeed a classic among the classics with contributions from several great mathematicians:
everything we need was elegantly presented in Segre's thesis on quadrics \cite{Segre}, he  introduced the modern notation (Segre symbols) in his classification of collineations and emphasised the geometrical ideas behind the classification;   the main algebraic concepts  (elementary divisors,normal forms) were developed in the context of the theory of determinants by   Weierstrass and  other members of the Berlin school  ( Kronecker, Frobenius); several of their  results  were obtained earlier by  Sylvester but in a less general and systematic way;  Sylvester classified nonsingular pencils of conics and quadric surfaces.  The modern reference on the classification of pencils of quadrics  is chapter XIII of the second volume of the book by Hodge and Pedoe \cite{HodgePedoe}.  More recently, Dimca has obtained a geometric interpretation of the classification of quadrics based on the geometry of determinantal varieties and their singularities \cite{Dimca}.

\subsection{Weierstrass models in F-theory}

Following its early founding papers \cite{Vafa:1996xn,Morrison:1996na,Bershadsky:1996nh},  in F-theory, elliptic fibrations are
traditionally studied  using Weierstrass models, i.e., a hypersurface in a $\mathbb{P}^2$-bundle over the Type-IIB base $B$ which in its reduced form is defined as the zero-scheme associated with the locus

\[
y^2z=x^3+fxz^2+gz^3,
\]
where $f$ and $g$ are sections of appropriate tensor powers of a line bundle $\mathscr{L}$ on $B$.
A smooth Weierstrass model admits a unique section and only two types of singular fibers: nodal curves (Kodaira fiber of type $I_1$) and cuspidal curves  (Kodaira fiber of type $II$).
It follows that in the world of Weierstrass models, singularities in the total space of the fibration  must be introduced to allow more interesting singular fibers to appear. Such singularities are necessary for example to describe non-Abelian gauge groups.
Any elliptic fibration endowed with a smooth section is birationally
equivalent to a (possibly) singular Weierstrass model \cite{Formulaire},
which might give the impression that one does not need to leave the world of  Weierstrass models when working on F-theory since
the section identifies the space $B$ as seen by type IIB string theory. However, there is much value in exploring non-Weierstrass
models in F-theory since the physics of F-theory is not invariant under birational transformations.
Not only that, but the M-theory approach to F-theory doesn't even require a section, so even elliptic fibrations
without a sections are physically relevant.
F-theory with  discrete fluxes and/or torsion can be naturally introduced by considering other models of elliptic curves than Weierstrass models\cite{KMT,CHL}. This usually requires a non-trivial Mordell-Weil group \cite{AspinwallMorrison} and can also be analyzed by considering special Weierstrass models, but their expressions are usually complicated.
For Weierstrass models, singular fibers over codimension-one loci in the base can be described using Tate's algorithm without
performing a systematic desingularization. The resulting fibers are those classified by Kodaira for singular fibers of an elliptic surface. Singular fibers above higher codimension loci are not
necessarily on Kodaira's list \cite{Miranda,Szydlo,MorrisonTaylor,EY,CCVG} and can even have components that jump in dimension \cite{Miranda,Codim3}.
The resolution of singularities of a Weierstrass model is not unique and different resolutions of the same singular Weierstrass model can have different types of singular fibers in higher codimension in the base\cite{Miranda,Szydlo,EY}.
Recently, this was shown to occur even for popular physically relevant models such as the $\SU(5)$ GUT \cite{EY}.
Considering other models of elliptic fibrations other than Weierstrass models allow the convenience of a rich spectrum of singular fibers without introducing singularities in the total space of the elliptic fibration \cite{AE2}. In this way, we can have F-theory descriptions of certain non-Abelian gauge theories without introducing any singularities into the total space of the fibration. The spectrum of singular fibers can be determined without any ambiguity. As explained in \cite{AE2}, elliptic fibrations not in Weierstrass form naturally admit weak coupling limits as well (analogous to Sen's orientifold limit of Weierstrass models), and provide descriptions of systems of seven-branes admitting a type IIB weakly coupled regime which consisting of super-symmetric brane-image-brane configurations that would be challenging to describe  in the traditional Weierstrass model approach to F-theory\cite{AE2}.

\begin{table}[hbt]
\begin{center}
\begin{tabular}{|c|c|c| c | l  |c|c|}
\hline
{\small Type }&  ${\rm ord}(F)$ & ${\rm ord}(G)$ & ${\rm ord}(\Delta)$

                                                                                          & $j$  &

                                                                                          {\footnotesize Monodromy}                & Fiber \\
                               \hline
$I_0$ & $\geq 0$  & $\geq 0$ &  {\small $0$} &
{\small
$\mathbb{R}$
}
 & $\mathrm{I}_2$ & Smooth torus \\
\hline
$I_1$ & $0$ & $0$ &  {\small $1$}  & $\infty$ &
$
\begin{pmatrix}
1& 1\\
0 & 1
\end{pmatrix}
$

 &
{  \setlength{\unitlength}{1 mm}
\begin{picture}(10,10)(-16,-4)
\put(-5,0){
\qbezier(-5,0)(4,6)(5,0)
\qbezier(-5,0)(4,-6)(5,0)
}
\put(-10,0){\qbezier(0,0)(-1,-1)(-3,-3)
\qbezier(0,0)(-1,1)(-3,3)
\put(-26,0){\footnotesize (Nodal curve)}
}
\end{picture}}

\\
\hline
$I_n$ &{\small $0$} &  {\small $0$} & {\small $n>1$ }  & $\infty$ &
{\small $
\begin{pmatrix}
1& n\\
0 & 1
\end{pmatrix}
$
}
 &
{  \setlength{\unitlength}{.9 mm}
\begin{picture}(35,13)(-5,-2)
\put(-6,6){\circle{4}}
\put(-1,1){\line(-1,1){4}}
\put(0,0){\circle{4}}
\put(2.2,0){\line(1,0){2}}
\put(6.4,0){\circle{4}}
\put(8,0){\line(1,0){2}}
\put(12,0){\circle{4}}
\put(14,0){\line(1,0){1}}
\put(16,0){\line(1,0){1}}
\put(18,0){\line(1,0){1}}
\put(20,0){\line(1,0){1}}
\put(22,0){\circle{4}}
\put(23,2){\line(1,1){4}}
\put(28,7){\circle{4}}
\put(-6.6,5.3){\tiny{$1$}}
\put(-.5,-.5){\tiny{$1$}}
\put(5.9,-.5){\tiny{$1$}}
\put(12,-.5){\tiny{$1$}}
\put(21.4,-.5){\tiny{$1$}}
\put(-.5,6.6){\tiny{$1$}}
\put(5.9,6.6){\tiny{$1$}}
\put(12,6.6){\tiny{$1$}}
\put(27.4,6.6){\tiny{$1$}}

\put(0,7){\circle{4}}
\put(2.2,7){\line(1,0){2}}
\put(6.4,7){\circle{4}}
\put(8,7){\line(1,0){2}}
\put(12,7){\circle{4}}
\multiput(14,7)(2,0){6}{\line(1,0){1}}
\put(-4,7){\line(1,0){2}}
\end{picture}}
 \\
\hline
$II$ & $\geq 1$ &  $1$ & $2$   &$0$&
{\small $
\begin{pmatrix}
1& 1\\
-1 & 0
\end{pmatrix}
$

}
&
{  \setlength{\unitlength}{1 mm}
\begin{picture}(10,10)(-10,-4)
\put(0,0){\qbezier(0,0)(2.7,.7)(5,5)
\qbezier(0,0)(2.7,-.7)(5,-5)
\put(-30,-.7){\footnotesize {Cuspidial curve}   }
}
\end{picture}}
 \\
\hline
$III$ & {\small$1$ } & {\small $\geq 2$ }  & {\small  $3$} &  {\small $1$} &
{\small $
\begin{pmatrix}
0 & 1\\
-1 & 0
\end{pmatrix}
$
}
  &
{  \setlength{\unitlength}{.7 mm}
\begin{picture}(10,10)(-2,-2.5)
\put(0,0){
\qbezier(0,5)(4.242,4.242)(5,0)
\qbezier(0,-5)(4.242,-4.242)(5,0)}
\put(10,0){\qbezier(0,5)(-4.242,4.242)(-5,0)
\qbezier(0,-5)(-4.242,-4.242)(-5,0)
}
\end{picture}}
 \\
\hline
$IV$ &  {\small $\geq 2$} & {\small $2$} & {\small $4$}  &    {\small $0$} &
{\small $
\begin{pmatrix}
0 & 1\\
-1 & -1
\end{pmatrix}
$
}
 &
{  \setlength{\unitlength}{1 mm}
\begin{picture}(10,10)(-5,-2)
\put(0,0){\qbezier(-2.5,-4.33)(0,0)(2.5,4.33)
\qbezier(-2.5,4.33)(0,0)(2.5,-4.33)
\qbezier(-5,0)(0,0)(5,0)
}
\end{picture}}
\\
\hline
$I^*_n$ & {\small $2$} & {\small $\geq 3$ } & {\small $n+6$} & {\small  $\infty$} &
{\small $
\begin{pmatrix}
-1& -b\\
0 &- 1
\end{pmatrix}
$
}
  &
{  \setlength{\unitlength}{.9 mm}
\begin{picture}(35,13)(-5,-2)
\put(-6,6){\circle{4}}
\put(-6,-6){\circle{4}}
\put(-1,-1){\line(-1,-1){4}}
\put(-1,1){\line(-1,1){4}}
\put(0,0){\circle{4}}
\put(2.2,0){\line(1,0){2}}
\put(6.4,0){\circle{4}}
\put(8,0){\line(1,0){2}}
\put(12,0){\circle{4}}
\put(14,0){\line(1,0){1}}
\put(16,0){\line(1,0){1}}
\put(18,0){\line(1,0){1}}
\put(20,0){\line(1,0){1}}
\put(22,0){\circle{4}}
\put(23,-1){\line(1,-1){4}}
\put(23,2){\line(1,1){4}}
\put(28,7){\circle{4}}
\put(28,-6){\circle{4}}
\put(-6.6,5.3){\tiny{$1$}}
\put(-6.6,-6.4){\tiny{$1$}}
\put(-.5,-.5){\tiny{$2$}}
\put(5.9,-.5){\tiny{$2$}}
\put(12,-.5){\tiny{$2$}}
\put(21.4,-.5){\tiny{$2$}}
\put(27.4,6.6){\tiny{$1$}}
\put(27.4,-6.6){\tiny{$1$}}
\end{picture}}
\\*
 \cline{2-4}  &  {\small $\geq 2$} &{\small  $3$} &{\small  $n+6$} &  
& &  \\
\hline
$IV^*$ & {\small  $\begin{matrix}\\   \geq 3\\  \\  \end{matrix}$} &{\small  $4$ }& {\small $8$ }&    {\small $0$
}&
{\small $
\begin{pmatrix}
-1& -1\\
1 & 0
\end{pmatrix}
$
}

&
\setlength{\unitlength}{1 mm}
\begin{picture}(25,16)(2,-10)
\put(0,0){\circle{4}}
\put(2.2,0){\line(1,0){2}}
\put(6.4,0){\circle{4}}
\put(8.6,0){\line(1,0){2}}
\put(12.6,0){\circle{4}}
\put(14.6,0){\line(1,0){2}}
\put(18.8,0){\circle{4}}
\put(21,0){\line(1,0){2}}
\put(25.2,0){\circle{4}}
\put(12.6,-6.2){\circle{4}}
\put(12.6,-12.4){\circle{4}}
\put(12.6,-8.4){\line(0,-1){2}}
\put(12.6,-2){\line(0,-1){2}}
\put(12.2,-13.4){\tiny{$1$}}
\put(24.7,-.5){\tiny{$1$}}
\put(12.2,-.5){\tiny{$3$}}
\put(-.5,-.5){\tiny{$1$}}
\put(5.9,-.5){\tiny{$2$}}
\put(18.3,-.5){\tiny{$2$}}
\put(12.2,-6.4){\tiny{$2$}}
\end{picture}

 \\
\hline
$III^*$ &
{\small  $\begin{matrix}\\   3\\   \\  \end{matrix}$ }& {\small  $\geq 5$} & {\small $9$}&  {\small $1$}
& {\small  $
\begin{pmatrix}
0& -1\\
1 & 0
\end{pmatrix}
$}
 &
\setlength{\unitlength}{1 mm}
\begin{picture}(35,10)(4,-4)
\put(0,0){\circle{4}}
\put(2.2,0){\line(1,0){2}}
\put(6.4,0){\circle{4}}
\put(8.6,0){\line(1,0){2}}
\put(12.6,0){\circle{4}}
\put(14.6,0){\line(1,0){2}}
\put(18.8,0){\circle{4}}
\put(21,0){\line(1,0){2}}
\put(25,0){\circle{4}}
\put(31.4,0){\circle{4}}
\put(37.8,0){\circle{4}}
\put(27.2,0){\line(1,0){2}}
\put(33.6,0){\line(1,0){2}}
\put(18.8,-6.2){\circle{4}}
\put(18.8,-2){\line(0,-1){2}}
\put(24.7,-.5){\tiny{$3$}}
\put(12.2,-.5){\tiny{$3$}}
\put(-.5,-.5){\tiny{$1$}}
\put(5.9,-.5){\tiny{$2$}}
\put(18.3,-.5){\tiny{$4$}}
\put(31,-.5){\tiny{$2$}}
\put(37,-.5){\tiny{$1$}}
\put(18.3,-7){\tiny{$2$}}
\end{picture}
\\
\hline
$II^*$ & {\small  $\begin{matrix}\\  \geq 4\\  \\    \end{matrix}$} & {\small  $5$} & {\small  $10$} & {\small  $0$ }&
{\small
$
\begin{pmatrix}
0& -1\\
1 & 1
\end{pmatrix}
$
}
&
\setlength{\unitlength}{.9 mm}
\begin{picture}(50,12)(-2,-6)
\put(0,0){\circle{4}}
\put(2.2,0){\line(1,0){2}}
\put(6.4,0){\circle{4}}
\put(8.6,0){\line(1,0){2}}
\put(12.6,0){\circle{4}}
\put(14.6,0){\line(1,0){2}}
\put(18.8,0){\circle{4}}
\put(21,0){\line(1,0){2}}
\put(25,0){\circle{4}}
\put(31.4,0){\circle{4}}
\put(37.8,0){\circle{4}}
\put(44,0){\circle{4}}
\put(27.2,0){\line(1,0){2}}
\put(33.6,0){\line(1,0){2}}
\put(40,0){\line(1,0){2}}
\put(31.4,-6.2){\circle{4}}
\put(31.4,-2){\line(0,-1){2}}
\put(24.7,-.5){\tiny{$5$}}
\put(12.2,-.5){\tiny{$3$}}
\put(-.5,-.5){\tiny{$1$}}
\put(5.9,-.5){\tiny{$2$}}
\put(18.3,-.5){\tiny{$4$}}
\put(31,-.5){\tiny{$6$}}
\put(37,-.5){\tiny{$4$}}
\put(43.5,-.5){\tiny{$2$}}
\put(31,-7){\tiny{$3$}}
\end{picture}

\\
\hline
\end{tabular}
\end{center}
\caption{ {  \footnotesize{\bf Kodaira Classification of  singular fibers of an elliptic fibration.}
The fiber of type $I^*_0$ is special among  its family $I^*_n$ because its $j$-invariant can take any value  in $\mathbb{C}$. The $j$-invariant of a   fiber of type $I_n$ or $I^*_n$ ($n>0$) has a pole of order $n$. }
}\label{table.KodairaTate}
\end{table}

\subsection{Other models of elliptic fibrations}
Since the physics of F-theory is not invariant under birational transformations, we would like to broaden our horizons and explore the landscape of F-theory beyond that of Weierstrass models. Our starting point is to consider the following families of elliptic curves \cite{Husemoller}:
\begin{align}\nonumber
\begin{tabular}{rll}
Weierstrass  cubic :& $zy^2=x^3+ f x z + g z^3$ & in $\quad\mathbb{P}^2$ \\
Legendre cubic: & $z y^2=x(x-z)(x-f z)$ & in $\quad\mathbb{P}^2$ \\
Jacobi quartic:&$y^2= x^4+ f x^2 z^2+  z^4 $  &  in $\quad\mathbb{P}^2_{1,2,1}$ \\
Hesse cubic: &$y^3+x^3-z^3- d  x y z =0$ & in $\quad\mathbb{P}^2$ \\
Jacobi intersection: & $ x^2-y^2- z^2=w^2- x^2 - d z^2=0$ &  in  $\quad\mathbb{P}^3$
\end{tabular}
\end{align}
where  $[x:y:z]$  (resp. $[x:y:z:w]$) are projective coordinates of a (weighted) $\mathbb{P}^2$ (resp. $\mathbb{P}^3$).
The coefficients ($f,g,d)$ in the equations above are scalars which we can interpret as sections of a line bundle over a point.
We then construct elliptic fibrations by promoting the coefficients of a particular family of elliptic curves to sections of line bundles over a base variety $B$.
We then consider the following normal forms of elliptic fibrations associated with the families of elliptic curves listed above \cite{KMT,CHL,KLRY,AE2}:
\begin{align}\nonumber
\begin{tabular}{llll}
$E_8$: & $y^2z=x^3+ f x z + g z^3$ & in & $\mathbb{P}[\mathscr{O}\oplus \mathscr{L}^2\oplus \mathscr{L}^3]$ \\
$E_7$: &$y^2= x^4+ f x^2 z^2+ g x z^3 + e z^4\quad $& in & $ \mathbb{P}_{1,1,2}[\mathscr{O}\oplus \mathscr{L}\oplus \mathscr{L}^2]$ \\
$E_6$:&$y^3+x^3=z^3+ d x y z +e x z^2 +f  y z^2 +g z^3$& in & $  \mathbb{P}[\mathscr{O}\oplus \mathscr{L}\oplus \mathscr{L}]$ \\
$D_5$: &$x^2-y^2- z( az+cw)=w^2-  x^2  -z (d z+e x+f y) =0 $ & in & $\mathbb{P}[\mathscr{O}\oplus \mathscr{L}\oplus \mathscr{L}\oplus \mathscr{L}]$
\end{tabular}
\end{align}
The $E_8$ family is the usual Weierstrass model.  A more general form of the Weierstrass model (the Tate form), will have the  fibration obtained from the Legendre family as a specialization.  The $E_7, E_6$ and $D_5$ elliptic fibrations are respectively obtained from generalizations of the Jacobi quartic,  the Hesse cubic and the Jacobi intersection form.
By promoting the scalar coefficients to sections of line bundles over a positive dimensional base variety $B$, we allow more  ``room" for \emph{singular} fibers to appear, and a richer
geometry naturally emerges. The $E_7$, $E_6$, and $D_5$ fibrations are all birationally equivalent to a singular Weierstrass model and the corresponding birational map is an isomorphism away from the locus of singular fibers.
Each model differs by the number of rational sections and the type of singular fibers it admits.
This $E_n$ nomenclature follows \cite{KMV,KLRY,AE2} and is based on an  analogy with del Pezzo surfaces\footnote{A del Pezzo surface of degree $d$ admits  $(-1)$-curves that define a root lattice of type $E_{9-d}$. A del Pezzo surface  of degree $d$ can be embedding in a projective space $\mathbb{P}^d$ as a surface of degree $d$. An hyperplane will cut such a del Pezzo surface along  an elliptic curve expressed as a degree $d$ curve.  A cone over an elliptic curve of type $E_n$ will have an elliptic  singularity of type $\tilde{E}_n$. A del Pezzo surface of degree $1,2$ and $3$ can be expressed as an hypersurface in a weighted projective $\mathbb{P}^3$ while a del Pezzo surface of degree $4$ can be expressed as a complete intersection of two quadric hypersurfaces in $\mathbb{P}^4$. The intersection with a
hyperplane gives the model discussed above.    The $E_8$ family is the usual Weierstrass model and the  $D_5$ family corresponds to $E_5=D_5$. }.
All these fibrations have been analyzed in \cite{AE2} with the exception of the  $D_5(\simeq E_5)$ elliptic fibration. By direct inspection of the results of \cite{AE2}, we observe the following:
\begin{prop}[Fiber geometry of $E_8$, $E_7$ and $E_6$ elliptic fibrations]
A general $E_{9-n}$ ($n=1,2,3$) elliptic fibration admits   $n$ sections and its spectrum of singular fibers contains  $2n$ different singular fibers, which are all the  Kodaira fibers composed of at most  $n$ irreducible rational curves.
\end{prop}

\subsection{$D_5$ elliptic fibrations}

$D_5$ elliptic fibrations have not received much attention in the physics literature.
This is mostly because the generic fiber of a $D_5$ fibration is a complete intersection while in the case of $E_8$, $E_7$, $E_6$ , it is simply a  hypersurface in a (weighted) projective plane.
In view of the properties of the $E_{9-n}$ elliptic fibrations for $n=1,2,3$, one might  expect that the $D_5=E_{9-4}$ elliptic fibration has an even  richer geometry. As we will see in this paper, an $D_5$ elliptic fibration  has 4 sections and admits 8 types of singular fibers composed of up to 4 components.
However, only 7 appear on the list of Kodaira. We will see that a general $D_5$ elliptic fibration with 4 sections indeed admits a non-Kodaira fiber composed of four rational curves meeting at a point. We call such a fiber a fiber of  type $I^{*-}_0$ since it looks like a Kodaira fiber $I^{*}_0$ with the central node contracted to a point.
 We study the physical significance of their non-Kodaira fibers by exploring weak coupling limits associated with them.

\begin{table}[htb]\label{En.Sing.Fib.}
\begin{center}
\begin{tabular}{|l|c|l|}
\hline
Type & sections & \multicolumn{1}{|c|}{ Singular fibers}\\
\hline
$E_8$ & 1 & $I_1$, $II$ \\
\hline
$E_7$ & 2 &  $I_1$, $II$, $I_2$, $III$\\
\hline
$E_6$ & 3 & $I_1$, $II$, $I_2$, $III$, $I_3$ , $IV$\\
\hline
$E_5=D_5$ & 4 & $I_1$, $II$, $I_2$,   $III$,  $I_3$, $IV$, $I_4$, $I^{*-}_0$ (non-Kodaira)\\
\hline
\end{tabular}
\end{center}
\caption{Singular fibers of an elliptic fibration of type $E_n$ with $(9-n)$ sections. We denote $E_5$ by $D_5$ as it is familiar with Dynkin diagrams.}
\end{table}

\subsection{Canonical form for a $D_5$ model with four sections}

In this section, we will introduce our canonical form for an elliptic fibration of type $D_5$ with four sections. 
We will ensure that the 4 sections are given by a unique divisor composed of 4 non-intersecting irreducible components. Each of these components is a Weil divisor and they are two by two disjoints so that the 4 sections define 4 distinct points on each fiber.

\subsubsection{Notation and conventions}
We work over the field $\mathbb{C}$ of complex numbers but everything we say is equally valid over an algebraically closed field $k$ of characteristic zero.
We denote by $\mathbb{P}^n$ the projective space of dimension $n$ over the field $\mathbb{C}$.
Given a line bundle $\mathscr{L}$, we denote its dual by $\mathscr{L}^{-1}$, its $n$-th tensorial power by
$\mathscr{L}^{n}$ and the dual of its $n$-th tensorial power by  $\mathscr{L}^{-n}$. Given a coherent sheaf $\mathscr{E}$, we denote by $\mathbb{P}^n[\mathscr{E}]$ the projective bundle of lines of $\mathscr{E}$.

\subsubsection{Canonical form for a $D_5$ elliptic fibration with four sections}
Let $B$ be a non-singular compact complex algebraic variety endowed with a line bundle $\mathscr{L}$. We consider the rank $4$
vector bundle  $$\mathscr{E}=\mathscr{O}_B\oplus \mathscr{L}\oplus \mathscr{L}\oplus \mathscr{L},$$ and its associated
projectivization\footnote{Here we take the projective bundle of \emph{lines} in $\mathscr{E}$.} $\pi: \mathbb{P}(\mathscr{E})\to B$. We denote the tautological line bundle of $\mathbb{P}(\mathscr{E})$ by
$\mathscr{O}(-1)$ and its dual by $\mathscr{O}(1)$. The vertical coordinates of $\mathbb{P}(\mathscr{E})$ are denoted by
$[x:y:z:w]$, where  $x$, $y$, $w$ are all  sections of $\mathscr{O}(1)\otimes \pi^{*}\mathscr{L}$ while  $z$ is a section of
$\mathscr{O}(1)$. We define a $D_5$ \emph{elliptic fibration} $Y$ to be a non-singular complete intersection determined
by the vanishing locus of two sections of $\mathscr{O}(2)\otimes \pi^{*}\mathscr{L}^2$. Such a complete intersection
determines an elliptic fibration $\varphi: Y\rightarrow B$, whose generic fiber is a complete intersection of two quadrics in
$\mathbb{P}^3$.
We also assume that the  elliptic fibration has a (multi) section cut out by  $z=0$.
It follows that the $D_5$ elliptic fibration $Y$ is given by:
\begin{equation}
Y:=
\begin{cases}
A_1(x,y,w)-z L_1(z,x,y,w)=0\\
A_2(x,y,w)-z L_2(z,x,y,w) =0
\end{cases}\nonumber
\end{equation}
where $A_1(x,y,w)$ and $A_2(x,y,w)$ denote two quadratic polynomials in $\mathbb{C}[x,y,w]$, while $L_1(z,x,y,w)$ and
$L_2(z,x,y,w)$ are linear in $x,y,z,w$ with coefficients that are sections of appropriate powers of $\pi^{*}\mathscr{L}$
so that each of $A_i-zL_i$  for $i=1,2$ is a section of $\mathscr{O}(2)\otimes \pi^{*}\mathscr{L}^2$.
We exclude the degenerate case where  $Q_1$ and $Q_2$ are proportional to each other.
It follows that fiberwise, the multisection cut out by $z=0$ defines up to four points on the elliptic fiber, corresponding to the fact each (distinct) solution the system $A_1(x,y,w)=A_2(x,y,w)=0$ determines a section of the elliptic fibration.
If $A_1$ and $A_2$ intersect transversally, we have  exactly  four  sections. We can also consider degenerate cases where the intersection is not transverse and would therefore lead to intersection points with multiplicities. Using non-transverse quadrics, we can have one, two or three sections\footnote{ For example $(A_1,A_2)=(x^2, w^2)$  gives a unique solution of multiplicity 4.
$(x^2-y^2, w^2)$ gives two solutions of multiplicity 2.
$(x^2-y^2, w^2-x^2+xy)$ gives three solutions, one of multiplicity one and the other of multiplicity two.}.

For the remainder of this article, unless otherwise mentioned we only consider the  case where the elliptic fibration admits exactly four sections.
In that case, without loss of generality, the $D_5$ elliptic fibration can be expressed as follows:
\begin{equation}\label{ci}
Y:=
\begin{cases}
x^2-y^2- z( az+cw)=0,\\
w^2-  x^2  -z (d z+e x+f y) =0.
\end{cases}
\end{equation}
This is our canonical form for a $D_5$ elliptic fibration with four rational sections.
So that each equation defines a section of $\mathscr{O}(2)\otimes \pi^{*}\mathscr{L}^2$,
we take $a$ and $d$ to be sections of $\pi^{*}\mathscr{L}^2$ and $c$, $e$ and $f$ to be sections of
$\pi^{*}\mathscr{L}$:
\begin{center}
\begin{tabular}{|c|c|c|c|}
\hline
$x,y, w$ & $z$ & $c, e, f$  & $a ,d$\\
\hline
$\mathscr{O}(1)\otimes \pi^{*}\mathscr{L}$ & $\mathscr{O}(1)$ & $\pi^{*}\mathscr{L}$ & $\pi^{*}\mathscr{L}^2$\\
\hline
\end{tabular}
\end{center}

\subsection{Pencil of quadrics}
To study the complete intersection $Y: Q_1= Q_2=0$ of two quadrics, it is useful to analyze the pencil of quadrics through $Y$.
It is  defined as follows
\begin{equation}
Q_{\lambda_1,\lambda_2}:\lambda_1 Q_1+\lambda_2 Q_2, \quad [\lambda_1:\lambda_2]\in \mathbb{P}^1.
\end{equation}
The variety $Y:Q_1=Q_2=0$ could equivalently be defined as the complete intersection
$\lambda_1 Q_1+ \lambda_2 Q_2=\mu_1 Q_1+\mu_2 Q_2=0$ for any choice of $\lambda_1, \lambda_2,\mu_1,\mu_2$ such that
$\lambda_1 \mu_2-\lambda_2 \mu_1\neq 0$. The curves $Q_1=Q_2=0$ is common to all the quadrics of the pencil. It is called the {\em base locus} of the pencil.
We denote the symmetric matrix corresponding to a quadric polynomial $Q=\sum a_{ij} x^i x^j$ as  $\hat{Q}$.
Singular fibers can be characterized by algebraic properties of the pencil.
In particular, the matrix $(\hat{Q}_1+r \hat{Q}_2)$ associated with the  pencil $Q_1+r Q_2$  has algebraic invariants known as {\em elementary divisors} that can be used to characterize uniquely a certain types of singular fibers. The elementary divisors are obtained from the  roots of the discriminant of the pencil and the common roots of  the  minors of order $1,2,\cdots, n$. We recall that the  minors of order $s$ of a matrix $M$ are obtained by taking the determinant of the matrices cut down  from $M$ by removing $s$ rows and $s$ columns. For a pencil of quadrics in $\mathbb{P}^3$, we will consider the first, second and third  minors. The determinant of the matrix $M$ can be seen as the unique minor of order zero.

\subsection{Discriminant of the elliptic fibration from the pencil of quadrics}
The complete intersection $Q_1=Q_2=0$ gives a regular elliptic curve if and only if the determinant of the quadratic form
$\hat{Q}_{1}+r \hat{Q}_2$  (with $r=\frac{\lambda_2}{\lambda_1}$) is non-identically zero and does not have multiple roots. In other words, we can compute the discriminant of the elliptic fibration $Y$  as the discriminant of the following quartic in $r$:
\begin{equation}
4\mathrm{det}(\hat{Q}_1+r \hat{Q}_2)=q_0 + 4 q_1  r+ 6 q_2 r^2+4 q_3  r^3+ q_4 r^4.
\end{equation}
 One can show that the the $D_5$ elliptic fibration determined by $Q_1=Q_2=0$ has Weierstrass form
 $y^2z=x^3 +F xz^2 +G z^3$ , where
\begin{equation}\label{GF}
F=-(q_0 q_4 -4 q_1 q_3 +3q_2^2), \quad  G=2(q_0 q_2 q_4 + 2q_1 q_2 q_3-q_2^3-q_0 q_3^2 -q_1^2 q_4).
\end{equation}
This Weierstrass model is the Jacobian of the $D_5$ elliptic fibration.
We then simply compute the discriminant and $j$-invariant as
\begin{equation}
\Delta=4 F^3+ 27 G^2,\quad j=1728 \frac{4 F^3}{\Delta}.
\end{equation}

\subsection{Birational equivalent $E_6$ model}
We now obtain a birationally equivalent formulation of the fibration in which the  generic fiber is a plane cubic curve. The plane cubic curve is obtained by projecting the space  curve on a plane from a rational point.
In order to proceed, we need to choose a rational point on every fiber of $Y$. For example, we can take the rational point $P=[1,1,1,0]$ which is one of the sections. We perform a translation $y\mapsto y+x , w\mapsto w+x$ so that in the new coordinate system, the point $P$ is $[1:0:0:0]$. It follows that there should be no terms  in $x^2$ in the defining equations. Indeed, after the substitution   $(y\mapsto y+x , w\mapsto w+x)$ in the defining equations of $Y$, we can eliminate $x$.
Geometrically, this is equivalent to projecting $Y$ to the plane $x=0$ from the point $P=[1:0:0:0]$. The result is the  following cubic:
\begin{align}
&(y^2+a z^2 + c w z)(2w+ez+fz)+(w^2-d z^2- f z y)(2 y + c z) =0.
\end{align}
where $[y,w,z]$ are the projective coordinates of the $\mathbb{P}^2$ defined by $x=0$. This cubic is a section of $\mathscr{O}(3)\otimes \mathscr{L}^3$ and  $z=0$ admits a  multisection $z=0$  of degree  3. Indeed,  $z=0$ cuts the cubic  along the following loci $$2  y w (y + w) =0,$$ of the $\mathbb{P}^1$ with projective coordinates $[y:w]$. This  corresponds to the points $[y,w,z]=[0:1:0]$, $[1:0:0]$ and $[1:-1:0]$ on the cubic curve. These points correspond to the  sections of the original $D_5$ elliptic fibration with the exception of the point $P$ used to define the projection. As the new elliptic fibration is defined by a divisor of class $\mathscr{O}(3)\otimes \mathscr{L}^3$ in the $\mathbb{P}^2$ bundle $\mathbb{P}[\mathscr{O}_B \oplus \mathscr{L}\oplus\mathscr{L}]$ and admits three sections, we recognize it as an  $E_6$ elliptic fibration.

We still have the same $j$-invariant and the same discriminant locus. However, the fiber structure has changed. For example (i)the non-Kodaira fiber $I^{*-}_0$ located at $a=c=d=e=f=0$ is now a Kodaira fiber of type $IV$. (ii) The  $I_4$ fiber at $a=c=e=4d-f^2=0$ is now a $I_2$ fiber composed of a conic and a secant.
\subsection{Birationally equivalent Jacobi quartic model}

An elliptic curve can also be modeled by the double cover of a $\mathbb{P}^1$ branched at four distinct points.
For that purpose, we can use a weighted projective plane $\mathbb{P}^2_{2,1,1}$ and write the equation  as
$$
y^2=Q_4(u,v),
$$
where $[y:u:v]$ are the projective coordinates of $\mathbb{P}^2_{2,1,1}$ with $y$ of weight $2$ and $u$ and $v$ of weight $1$ and  $P_4(u,v)$ is homogeneous of degree $4$ in $[u:v]$.
The quartic $Q_4$ is simply given by the binary quartic polynomial determined by the polynomial  of the pencil of quadrics defining the $D_5$ elliptic fibration, so the expression is
\begin{equation}
y^2 =det(u \hat{Q}_1+ v \hat{Q}_2),
\end{equation}
which gives
\begin{equation}
y^2 =q_0 u^4+ 4 q_1 u^3 v+6 q_2 u^2 v^2+ 4 q_3 u v^3+ q_4 v^4.
\end{equation}
This elliptic fibration is given by a section of 
 $\mathscr{O}(4)\otimes\mathscr{L}^2$ written in the projective bundle $\mathbb{P}[\mathscr{O}_B\oplus \mathscr{O}_B\oplus\mathscr{L}]$. The projective variable  $y$ is a section of $\mathscr{O}(2)\otimes\mathscr{L}$ while  $u$ and $v$ are sections of $\mathscr{O}(1)$.
Since the generic fiber is modeled by a quartic in $\mathbb{P}^2_{2,1,1}$, we have a $E_7$ model.  However, compare to the $E_7$ fibrations that only had fibers of type $I_1, II, I_2$ and $III$, this variant of the $E_7$ elliptic fibration also admits a non-Kodaira fiber composed of a rational curve of multiplicity 2. This fiber is located at $q_0=q_1=q_2=q_3=q_4=0$.  in the $D_5$ elliptic fibration, as we will see later, this fiber would be the non-Kodaira fiber $I^{*-}_0$ composed of four rational curves meeting at a common point. 
The singular fibers can easily be classified by analyzing the factorization of $Q_4$ as reviewed in table \ref{E7.fib}. 
For another application of quartic elliptic curves in F-theory see \cite{Braun:2011ux}.

\begin{table}
\begin{center}
\begin{tabular}{|c|c|p{4cm}|}
\hline
Type & General condition &  Description \\
\hline
$I_1$ & $\Delta=0$ & A nodal curve ($Q_4$ has one double root) \\
\hline
$II$ & $F=G=0$ & A cuspidial curve ($Q_4$ has a triple root) \\
\hline
$I_2$ & $q_4 q_1^2-q_3^2 q_0=2 q_3^3+q_4^2 q_1-3  q_4 q_3 q_2=0$ & A tacnode($Q_4$ has two double root)\\ 
\hline
$III$  & $rank \begin{pmatrix} q_0 & q_1 & q_2 & q_3 \\ q_1 & q_2 & q_3 & q_4\end{pmatrix}\leq 1$ & Two conics tangent at a point ($Q_4$ has a quadruple root)\\
\hline
$2T_1$ & $q_0=q_1=q_2=q_3=q_4=0$ & A rational curve (in this case a projective line) of multiplicity 2\\
\hline
\end{tabular}
\end{center}
\caption{ Singular fibers of the elliptic fibration $y^2=q_0 u^4+ 4 q_1 u^3 v+6 q_2 u^2 v^2+ 4 q_3 u v^3+ q_4 v^4$ birationally equivalent to a $D_5$ elliptic fibration \label{E7.fib}}
\end{table}

Interestingly, if we introduce $[X_0:X_1:X_2:X_3]$ as projective coordinates of a $\mathbb{P}^3$, the  weighted projective space $\mathbb{P}^2_{1,1,2}$ is isomorphic  to the cone $X_1 X_2=X_0^2$  in $\mathbb{P}^3$. The explicit isomorphism is the following: 
\begin{equation}
[u:v:y]\mapsto [X_0:X_1:X_2:X_3]=[uv:u^2: v^2  : y].
\end{equation}
If we use this map starting from the  projective bundle $\mathbb{P}_{1,1,2}[\mathscr{O}_B\oplus \mathscr{O}_B\oplus \mathscr{L}^2]$, we get the following projective bundle $\mathbb{P}[\mathscr{O}_B\oplus \mathscr{O}_B\oplus \mathscr{O}_B\oplus \mathscr{L}^2]$.
 We can write it again as a $D_5$ elliptic fibration with one constant quadric as a section of $\mathscr{O}(2)$ and $\mathscr{O}(2)\otimes \mathscr{L}^4$: 
\begin{equation}
\begin{cases}
X_1 X_2- X_0^2=0,\\
X_3^2- q_0 X_1^2  - 4 q_1 X_0 X_1 - 6 q_2  X_0^2 -4 q_3 X_0 X_2    - q_4  X_2^2=0.
\end{cases}
\end{equation}
In this expression, we have a $\mathbb{Z}_2$ involution $\sigma:Y\rightarrow Y: X_3\mapsto -X_3$.

\subsection{Classification of singular fibers by Segre symbols}
We classify the singular fibers of a smooth $D_5$ elliptic fibration by using the classification of pencils of quadrics by Segre symbols.
When the discriminant $\mathrm{det}(\hat{Q}_1+r \hat{Q}_2)$ is not identically zero, we have a non-degenerate pencil of quadrics in $\mathbb{P}^3$. There are  14 different Segre symbols: one corresponds to a smooth elliptic curve,  nine correspond to seven singular fibers of Kodaira type  and four correspond to non-Kodaira fibers.
When the discriminant is identically zero, we have a pencil of quadrics in $\mathbb{P}^2$, which admit six different cases all corresponding to non-Kodaira fibers given by four lines meeting at a common point. We have described this case in table \ref{Class.Pencil.P3}. When the discriminant is identically zero as well as all the first order minors, we have a pencil in $\mathbb{P}^1$. This gives  3 additional singular fibers in higher dimension. Once we have a fibration with a certain number of sections, we have more constraints on the type of singular fibers that can occur. In particular, a smooth  $D_5$ elliptic fibration with four sections  admits eight different types of singular fibers, including one which doesn't appear on Kodaira's list and consists of four lines meeting at a common point. We will denote such a fiber by $I^{*-}_0$ since it looks like  a Kodaira fiber of type $I^*_0$ with the central node contracted to a point.  We will give a detailed classification of all possible fibers of our canonical $D_5$ model (with four sections) as  general conditions on the sections  $a,c,d,e,f$.  The classification by Segre symbols is based on the property of the matrix associated with the pencil of quadrics, in particular, the multiplicity of the zeros of its  consecutive minors. This will be the subject of section \ref{Section.Segre}. For the case with four sections, the list of fibers is given by the following:

\begin{prop}[Fiber geometry of  a $D_5$ elliptic fibration]
A general $D_5$ elliptic fibration whose generic fiber is defined by the  complete intersection of two quadrics in $\mathbb{P}^3$
admits four sections and 8 singular fibers. Namely the Kodaira fibers $I_1,I_2$, $III$, $I_3$, $IV$, $I_4$ and a non-Kodaira fiber  $I^{*-}_0$ composed of 4 lines meeting at a point.
\end{prop}

\begin{table}[hbt]
\begin{center}
\begin{tabular}{|p{1.6cm} |p{4.3cm}      |  p{6.8 cm} |}
\hline
Segre symbol & Roots of $\Delta$ and rank of associated quadric & {\hspace{2cm} Geometric             description} \\
\hline
$[1111]$ & 4 simple roots & smooth quartic ( $I_0$)\\
\hline
$[112]$ & one double root, rank 3 & nodal quartic  ($I_1$) \\
\hline
$[11(11)]$ & one double root, rank 2 & two intersecting conics  ($I_2$) \\
\hline
$[13]$ & triple root, rank 3& cuspidial quartic ($II$) \\
\hline
$[1(21)]$ & triple root, rank 2 & two tangent conics  ($III$)\\
\hline
$[1(111)]$ & triple root, rank 1 &  double conic  \\
\hline
$[4]$ & quadruple root, rank 3 & cubic and tangent line  ($III$)\\
\hline
$[(31)]$ & quadruple root, rank  2 & {\footnotesize conic and 2 lines  meeting on the conic ( $IV$)} \\
\hline
$[(22)]$ & quadruple root, rank 2& two lines and a double line \\
\hline
$[(211)]$ &quadruple root, rank 1  & two double lines \\
\hline
$[(1111)]$ & {\footnotesize quadruple root, rank 0} & The two quadrics coincide  \\
\hline
$[22]$ & 2 {\footnotesize double roots,both rank 3 }  & cubic and secant line ($I_2$)\\
\hline
$[2(11)]$ &{\footnotesize 2  double roots, rank 3 and 2}& {\footnotesize a conic and two lines forming a triangle($I_3$) }\\
\hline
$[(11)(11)]$   &{\footnotesize 2 double roots, both rank 2} &  four lines forming a quadrangle($I_4$) \\
\hline
\end{tabular}
\end{center}
\caption{Classification of  non-degenerate pencils of  quadrics in $\mathbb{P}^3$. In the second column, $\Delta$ is the discriminant of the pencil of quadrics. In the last column,  when the fiber is in Kodaira list, we mention its Kodaira symbol in parenthesis.  \label{Class.Pencil.P3}}
\end{table}

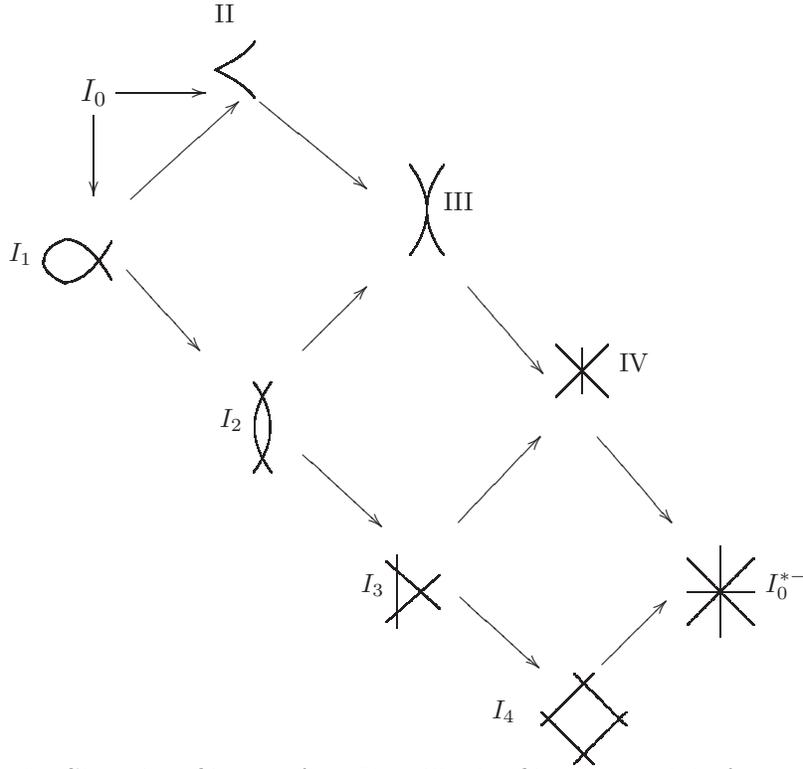
\begin{figure}[htb]
\begin{center}
\vspace{1cm}
\begin{tabular}{c}
\xymatrix{
I_0 \ar[r]\ar[d]& {
\setlength{\unitlength}{.15cm}
\begin{picture}(6, 0)(0,-2)
\qbezier(0,0)(2.8,1.3)(3.5,2.5)
\qbezier(0,0)(2.8,-1.3)(3.5,-2.5)
\put(0,4.2){{\footnotesize II}}
\end{picture}
}
 \ar[rd] & & & \\
 {
\setlength{\unitlength}{.15cm}
\begin{picture}(5,5)
\qbezier(0,1.9)(-1.9,1.2)(-1.9,0)
\qbezier(0,-2)(-1.9,-1.2)(-2,0)
\qbezier(0,1.9)(2,1.9)(4.1,-1.7)
\qbezier(0,-1.9)(2,-1.9)(4.1,1.7)
\put(-5,0){\footnotesize $I_1$}
\end{picture}
}\ar[dr]\ar[ur]& & {
 \setlength{\unitlength}{.3cm}
\begin{picture}(4, 4)(-3,-3)
\qbezier(0,2)(-1.5,0)(0,-2)
\qbezier(-1.5,2)(0,0)(-1.5,-2)
\put(0,0){\footnotesize III}
\end{picture}
}\ar[dr] &  & \\
&  {
\setlength{\unitlength}{.3cm}
\begin{picture}(4,4)(-3,-1)
\qbezier(0,2)(-1.5,0)(0,-2)
\put(-.8,0){\qbezier(0,2)(1.5,0)(0,-2)}
\put(-2.3,0){\footnotesize  $I_2$}
\end{picture}
} \ar[ur] \ar[dr]&  & {
\setlength{\unitlength}{.15cm}
\begin{picture}(4, 4)(-3,-5)
\qbezier(0,-2)(0,0)(0,2)
\qbezier(2.3,2.3)(0,0)(-2.3,-2.3)
\qbezier(2.3,-2.3)(0,0)(-2.3,2.3)
\put(3.3,0){\footnotesize IV}
\end{picture}
}\ar[dr] & \\
& & {
\setlength{\unitlength}{.2cm}
\begin{picture}(4, 4)(-2,0)
{\qbezier(1.2,1.1)(0,0)(-2.3,-2)}
{\qbezier(1.2,-1.1)(0,0)(-2.3,2)}
\qbezier(-1.6,-2.4)(-1.6,0)(-1.6,2.4)
\put(-4,0){\footnotesize  $I_3$}
\end{picture}
}\ar[ur]\ar[dr] &  & {
\setlength{\unitlength}{.2cm}
\begin{picture}(4, 4)(-3,0)
\qbezier(0,-3)(0,0)(0,3)
\qbezier(-2.2,0)(0,0)(2.2,0)
\qbezier(2.2,2.2)(0,0)(-2.2,-2.2)
\qbezier(2.2,-2.2)(0,0)(-2.2,2.2)
\put(3,0){\footnotesize  $I_0^{*-}$}
\end{picture}
} \\
& & &
{
\setlength{\unitlength}{.15cm}
\begin{picture}(4, 4)
\qbezier(-.6,-.6)(2,2)(4,4)\qbezier(-.6,.6)(2,-2)(4,-4)
\put(6.3,0){\qbezier(.6,-.6)(-2,2)(-4,4)\qbezier(.6,.6)(-2,-2)(-4,-4)}
\put(-5,0){\footnotesize
  $I_4$
}
\end{picture}
}
\ar[ur] &
}
\end{tabular}
\end{center}
\caption{Singular fibers of a $D_5$ elliptic fibration with four sections. There are a total of 8 singular fibers. This includes all the Kodaira fibers with at most 4 components and the   fiber $I^{*-}_0$ which is not  on Kodaira's list.
Down  arrows  represent an increase in the number of components while  up arrows indicate
a specialization  from  a semi-stable to an unstable fiber  while preserving the number of components of the fiber.
\label{figure.enhancement}
}
\end{figure}

\begin{table}[htb]
\begin{center}
\begin{tabular}{|c|p{7cm}| p{6cm}|}
\hline
& & \\
Type & \multicolumn{1}{|c|}{  \large General conditions   } & \multicolumn{1}{|c|}{\large Descriptions} \\
& &\\
\hline
$I_1 $ & $\quad \quad \quad \Delta=0$ & Nodal quartic $[211]$\\
\hline
 $II$ & $\quad \quad \quad  F=G=0$& Cuspidial quartic $[13]$\\
 \hline
\begin{tabular}{c}  \\ \\   $I_2$ \\ \end{tabular} & \begin{tabular}{ll}  & $ a=c=0$\\ or &  $f=4d-e^2=0$\\ or &   $e=4(a+d)+c^2-f^2=0$ \end{tabular}&
\begin{tabular}{p{4.9cm}}Two conics intersecting at two distinct points $[11(11)]$\end{tabular}  \\
 \cline{2-3}
 & \begin{tabular}{c} \\ $q_4 q_1^2-q_3^2 q_0 =2 q_3^3+q_4^2 q_1 - 3 q_4 q_3 q_2=0$\\  \\  \end{tabular}  & A twisted cubic and a secant $[22]$\\
 \hline
\begin{tabular}{c}   \\  \\   $III$ \end{tabular} & \begin{tabular}{p{.1cm}l}  & $ a=c=d=0$\\ or &  $f=4d-e^2=4a-e^2=0$\\ or &  {\small $e=4a +2c^2+f^2=4d-c^2-2f^2=0$} \end{tabular}&
\begin{tabular}{p{5cm}}Two tangent conics   $[1(21)]$\end{tabular}  \\
 \cline{2-3}
  &  $rank {\begin{pmatrix} q_0 & q_1 & q_2 & q_3 \\ q_1 & q_2 & q_3 & q_4  \end{pmatrix} } \leq 1$   & A twisted cubic and a tangent $[4]$\\
 \hline
\begin{tabular}{c}   $I_3$ \\ \end{tabular} & \begin{tabular}{ll}  & $ a=c=2q_3^2-3 q_2 q_4=0$\\ or &  $f=4d-e^2=2 q_1^2-3 q_0 q_2=0$\\ or &
$\begin{cases} e=4(a+d)+c^2-f^2=0\\
  (f^2-4d)^2-4 c^2 f^2=0 \end{cases}$ \end{tabular}&
\begin{tabular}{p{5.3cm}} A conic and two lines meeting as a triangle    $[2(11)]$\end{tabular}  \\
  \hline
\begin{tabular}{c}    $IV$ \\ \end{tabular} & \begin{tabular}{ll}  & $ a=c=d=e^2-f^2=0$\\ or &  $f=d-a=4d-e^2=4a-c^2=0$\\ or &   $\begin{cases}e=0\\ 4d=-4a=3c^2=3f^2\end{cases}$ \end{tabular}&
\begin{tabular}{p{4.9cm}} A conic meeting two lines at the same point $[(31)]$\end{tabular}  \\
\hline
\begin{tabular}{c}    $I_4$ \\ \end{tabular} & \begin{tabular}{ll}  & $ a=c=ef=4d-e^2-f^2=0$\\ or  &   $4a +c^2=d=e=f=0$ \end{tabular}&
\begin{tabular}{p{4.9cm}} four lines forming a quadrangle  $[(11)(11)]$\end{tabular}  \\
\hline
\begin{tabular}{c}    $I_0^{*-}$ \\ \end{tabular} & \begin{tabular}{ll}  & $ a=c=d=e=f=0$ \end{tabular}&
\begin{tabular}{p{4.9cm}} four lines meeting at a point  \end{tabular}  \\
\hline
\end{tabular}
\end{center}
\caption{Singular fibers of a $D_5$ elliptic fibration with four sections. We use the canonical form given in equation \eqref{ci} and $q_0,q_1, q_2, q_3 , q_4$ are defined in equation \eqref{qdef} and computed in equation \eqref{qi} : $q_0=c^2$,  $q_1=\frac{1}{4}(4a-c^2)$,    $q_2= \frac{2}{3}(d-a)$, $q_3=\frac{1}{4}(-4 d -f^2+e^2)$ and $q_4= f^2$.}
\end{table}

\subsection{Non-Kodaira fibers}
An attractive feature of the F-theory picture is that it proposes an  elegant dictionary between singular fibers and physical properties of type-IIB compactifications. The dictionary is well understood in codimension-one in the base where singular fibers determine the gauge group of the gauge theory living on the seven-branes.  More work needs to be done to understand the meaning of the matter representations and Yukawa couplings.
In the road to a better understanding of the physics of F-theory, one cannot hide away from  non-Kodaira singular fibers. As shown in \cite{EY}, non-Kodaira fibers can show up very naturally in important  models like for example the $\SU(5)$ Grand Unified Theory.
The physical meaning of non-Kodaira fibers can be explored in many different ways.
One can ask how they modify the matter content and the Yukawa couplings of the gauge theory associated with a given elliptic fibration.
This is the road explored recently by Morrison-Taylor \cite{MorrisonTaylor} in the context of F-theory on Calabi-Yau threefolds and by Marsano-Schafer-Nameki in the context of the small resolution of the $\SU(5)$ model \cite{Marsano:2011hv}.
It is also worthwhile to investigate weak coupling limits of F-theory in presence of non-Kodaira fibers.
 A general $D_5$ elliptic fibration can admit many possible non-Kodaira fibers. Some are higher dimensional fibers such as two quadric surfaces coinciding. The non-Kodaira fibers that are one dimensional are presented in figure \ref{NKfig}.

\begin{figure}[!htp]
\begin{center}
\begin{tabular}{ccccccc}
\setlength{\unitlength}{.2cm}
\begin{picture}(6, 10)(-3,-5)
\qbezier(-2.2,1)(0,-6.2)(2.2,1)
\put(2.4, -1){\footnotesize $2$}
\put(-2,-5){\footnotesize  $[1(111)]$}
\end{picture}
&
\setlength{\unitlength}{.2cm}
\begin{picture}(6, 10)(-3,-5)
\qbezier(-3,0)(0,0)(3,0)
\qbezier(-2,-2.2)(-2,0)(-2,2.2)
\qbezier(2,-2.2)(2,0)(2,2.2)
\put(-0.2,.8){\footnotesize $2$}
\put(2.4,2){\footnotesize $1$}
\put(-3.1,2){\footnotesize $1$}
\put(-2,-5){\footnotesize  $[(22)]$}
\end{picture}
&
\setlength{\unitlength}{.2cm}
\begin{picture}(6, 10)(-3,-5)
\qbezier(0,-3)(0,0)(0,3)
\qbezier(-2.2,0)(0,0)(2.2,0)
\qbezier(2.2,2.2)(0,0)(-2.2,-2.2)
\qbezier(2.2,-2.2)(0,0)(-2.2,2.2)
\put(0.2,2.3){\footnotesize $1$}
\put(2.4,2.4){\footnotesize $1$}
\put(2.3,0){\footnotesize $1$}
\put(2.4,-2.4){\footnotesize $1$}
\put(-2,-5){\footnotesize  $[111]$}
\end{picture}
&
\setlength{\unitlength}{.2cm}
\begin{picture}(6, 10)(-3,-5)
\qbezier(0,-3)(0,0)(0,3)
\qbezier(2.2,2.2)(0,0)(-2.2,-2.2)
\qbezier(2.2,-2.2)(0,0)(-2.2,2.2)
\put(0,2.3){\footnotesize $1$}
\put(2.4,2.4){\footnotesize $2$}
\put(2.4,-2.4){\footnotesize $1$}
\put(-2,-5){\footnotesize  $[21]$}
\end{picture}
&
\setlength{\unitlength}{.2cm}
\begin{picture}(6, 10)(-3,-5)
\qbezier(0,-3)(0,0)(0,3)
\qbezier(-2.2,0)(0,0)(2.2,0)
\put(.3,2){\footnotesize $2$}
\put(2,-1.5){\footnotesize $2$}
\put(-2,-5){\footnotesize  $[(11)1]$}
\end{picture}
&
\setlength{\unitlength}{.2cm}
\begin{picture}(6, 10)(-3,-5)
\qbezier(0,-3)(0,0)(0,3)
\qbezier(-2.2,0)(0,0)(2.2,0)
\put(.2,2){\footnotesize $1$}
\put(2,-1.5){\footnotesize $3$}
\put(0,-5){\footnotesize  $[3]$}
\end{picture}
&
\setlength{\unitlength}{.2cm}
\begin{picture}(6, 10)(-3,-5)
\qbezier(0,-3)(0,0)(0,3)
\put(.3,2){\footnotesize $4$}
\put(-2,-5){\footnotesize  $[(21)]$}
\end{picture}\\
\end{tabular}
\end{center}
\caption{One dimensional non-Kodaira fibers appearing in $D_5$ elliptic fibrations.\label{NKfig}}
\end{figure}
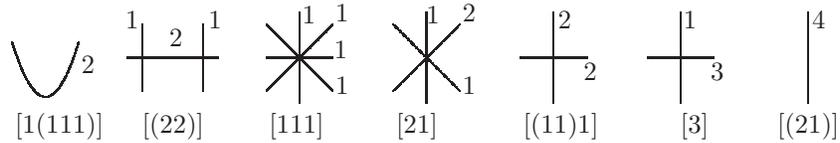

\subsection{Orientifold limits of $D_5$ elliptic fibrations}
The weak coupling limit of F-theory was first introduced by Sen \cite{Sen.Orientifold} in the case of a smooth Weierstrass model. Sen's limit gives a $\mathbb{Z}_2$ type IIB orientifold theory. Weak coupling limits for $E_6$ and $E_7$ elliptic fibrations were obtained in \cite{AE2} where a geometric description of the limit was also presented: a weak coupling limit is simply defined by a transition between a semi-stable fiber and an unstable fiber (semi-stable fibers admit an infinite $j$-invariant while the $j$-invariant of an unstable fiber is ``$\frac{0}{0}$" and so is undefined).

Since $D_5$ elliptic fibrations admit a wide variety of singular fibers, we then have many possible ways in which to explore weak coupling limits. A simple example of a weak coupling limit for $D_5$ elliptic fibrations can be obtained by considering the transition $I_2\rightarrow III$. We realize the weak coupling limit associated with $I_2\rightarrow III$ via the following family:
\begin{align}
Y_{\epsilon} (I_2\rightarrow III):
\begin{cases}
x^2-y^2-\epsilon z( \chi z+ \eta w)=0 \\
w^2-x^2-z\big[hz+(\phi_1+\phi_2)x+(\phi_1-\phi_2)y\Big]=0.
\end{cases}
\end{align}
The discriminant and $j$-invariant then take the following form at leading order in $\epsilon$:
\begin{equation}
\Delta\sim \epsilon^2 h^2(h-\phi_1^2)(h-\phi_2^2)(h\eta^2-\chi^2), \quad j\sim \frac{h^4}{\epsilon^2 (h-\phi_1^2)(h-\phi_2^2)(h\eta^2-\chi^2)}
\end{equation}
It is easy to see that at $\epsilon=0$, the first quadric splits into two planes. Each of these two planes will cut the second quadric along a conic.
The two conics intersects at two distinct points defining in this way a Kodaira fiber of type $I_2$. Such a fiber is semi-stable and admits an infinite value for the $j$-invariant.  At $\epsilon=h=0$, the two conics are tangent to each other and therefore define a Kodaira fiber of type $III$, which is an unstable fiber with an undefined  $j$-invariant of type $``\frac{0}{0}"$. After a glance at the $j$-invariant and discriminant it is immediately  clear that at $h=0$, we have an orientifold \cite{Sen.Orientifold, AE2}.
Taking the double cover $\zeta^2=h$, we see that the other components  $h-\phi_1^2$, $h-\phi_2^2$ and $h\eta^2-\chi^2$ of the discriminant split into brane-image-brane pairs in the double cover wrapping smooth loci mapped to each other by the $\mathbb{Z}_2$ involution $\zeta\rightarrow \zeta$.  All together we have one orientifold and 3 brane-image-brane pairs  wrapping smooth divisors:
$$\text{Brane spectrum at weak coupling}:\quad
\begin{cases}
{O}\phantom{1\pm}    \text{(orientifold)}&:\zeta=0\\
{D}_{1\pm}\  \text{(brane-image-brane)}&:\phi_1\pm \zeta=0\\
{D}_{2\pm}\  \text{(brane-image-brane)}&:\phi_2\pm \zeta=0\\
{D}_{3\pm}\   \text{(brane-image-brane)}&: \chi\pm \zeta\eta=0
\end{cases}
$$
We note that the orientifold $O$ and the brane-image-brane $D_{1\pm}$ and $D_{2\pm}$ are all in the same homology class: $[{O}]=[{D}_{1\pm}]=[{D}_{2\pm}]$.
 The orientifold limit we present corresponds to the transition  $I_2\rightarrow III$ when the brane-image-brane do not coincide with the orientifold.
One can think of each $\phi_i$ ($i=1,2$) as a modulus controlling the separation between the brane $D_{i+}$ and its image $D_{i-}$.
When $\phi_i=0$, $D_{i\pm}$ coincides with the orientifold. If we specialize to the case $\phi_1=\phi_2=0$ we obtain the following family:
\begin{align}
Y_{\epsilon} (I_2\rightarrow I^{*-}_0):
\begin{cases}
x^2-y^2-\epsilon z( \chi z+ \eta w)=0, \\
w^2-x^2-h z^2=0.
\end{cases}
\end{align}
The discriminant and $j$-invariant then take the following form at leading order in $\epsilon$:
\begin{equation}
\Delta\sim \epsilon^2 h^4(h\eta^2-\chi^2), \quad j\sim \frac{h^2}{\epsilon^2 (h\eta^2-\chi^2)}
\end{equation}
Both brane-image-brane pairs $D_{i\pm}$ coincide with the orientifold. Interestingly, in that case, the fiber above $h=0$ when $\epsilon=0$ is not of type $III$ (two rational curves meeting at a double point) but  become the non-Kodaira fiber $I^{*-}_0$ (four lines meeting at a point).

In both cases $I_2\rightarrow III$ and $I_2\rightarrow I^{*-}_0$, since $[O]=[D_{1\pm}]=[D_{2\pm}]=L$ and $[D_{3+}]=[D_{3-}]=L^2$ we expect a universal tadpole relation of the form \cite{AE2}
$$
\varphi_*c(Y)=\rho_* \big[4 c(O)+c(D_{3+})\big].
$$
This relation is  proven in section \ref{tadpole}. Taking the integral of both sides of the Chern class identity above immediately yields the numerical relation predicted by tadpole matching between type IIB and F-theory:
$$
\chi(Y)=4\chi(O)+\chi(D_{3+}).
$$
When $Y$ is a Calabi-Yau fourfold, this relation ensures that the D3 brane tadpole has the same curvature contribution in F-theory and in the type IIB weak coupling limit.

\subsection{Euler characteristic}
In F-theory, a Sethi-Vafa-Witten formula is an expression of the Euler characteristic of an elliptic fibration in terms of the topological numbers of its base.
Such formulas are particularly useful in the context of F-theory compactified on a Calabi-Yau elliptic fourfolds since the Euler characteristic of the fourfold enters the formula for the $D_3$ tadpole.
The first example of such a formula was actually obtained by  Kodaira for an elliptic surface. Sethi, Vafa and Witten computed the Euler characteristic of a  Calabi-Yau fourfold in the cases of an $E_8$ elliptic fibration over a smooth base\cite{SVW}:
$$
\text{Sethi-Vafa-Witten}: \quad \chi(Y)=12 c_1(B)c_2(B) + 360 c_1^3(B).
$$
Klemm-Lian-Roan-Yau have obtained general results for Calabi-Yau elliptic fibrations of type $E_n$ ($n=8,7,6$) over a base of dimension $d$ \cite{KLRY}.
Aluffi and Esole have obtained more general relations  without assuming the Calabi-Yau conditions for $E_n$ ($n=8,7,6$) elliptic fibrations of arbitrary dimension \cite{AE2}. These relations  express the simple geometric fact that the Euler characteristic of the elliptic fibration is a simple multiple of the Euler characteristic of a hypersurface in the base. The following theorem is proven in \cite{AE2}:
\begin{theorem}
Let $\varphi:Y\rightarrow B$ be an elliptic fibration of type $E_n$ ($n=8,7,6$). Such an  elliptic fibration is the zero section of a line bundle $\mathscr{O}(m)\otimes \pi^{*}\mathscr{S}$ where $\mathscr{S}$ is line bundle on the base and $\phi:\mathbb{P}(\mathscr{E})\rightarrow B$ is the (weighted) projective bundle in which the defining equation is written. Then
$$
\varphi_* c(Y)= (10-n) \iota_* c(Z),
$$
where $Z$ is  a smooth hypersurface in the base defined as the zero locus of a section of the line bundle $\mathscr{S}$ and $\iota:Z\rightarrow B$ is the embedding of $Z$ in $B$.
\end{theorem}
Sethi-Vafa-Witten formulas are then immediately obtained by computing the push-forward of the total Chern class of $Y$. A similar formula can be written in great generality for a fibration with generic fiber a plane curve of degree $d$ where the total space of the fibration is a hypersurface in a $\mathbb{P}^2$ bundle \cite{James}, see also \cite{CCVG}.
For $D_5$ elliptic fibrations we prove the following result:
\begin{theorem} Let  $\varphi:Y\rightarrow B$ be a $D5$ elliptic fibration, then
\begin{align}
&\varphi_*c(Y)=\frac{4L(3+5L)}{(1+2L)^2}c(B).\nonumber\\
&\chi(Y)= -\sum_{k=1}^{d} (-2)^k (5+k)L^k c_{d-k}(B),\quad d=dim\ B\nonumber
\end{align}
where $L=c_1(\mathscr{L})$ and $\chi(Y)$ denotes the  topological Euler characteristic of $Y$.
\end{theorem}
In particular, if the $D_5$ elliptic fibration is a Calabi-Yau fourfold, we recover the result of Klemm-Lian-Roan-Yau \cite{KLRY} for the Euler characteristic of a $D_5$ elliptically fibered Calabi-Yau fourfold:
$$
\chi(Y)=12 c_1(B) c_2(B)+36 c_1^3(B).
$$

\section{Geometry of quadric surfaces}

In this section,  we review some basic facts about the geometry of quadric surfaces. We will also describe the irreducible curves in such surfaces.  We will pay a special attention to the degeneration of an elliptic curve in a quadric surface.
Some important transitions that we want to describe are the degenerations of an elliptic curve into two conics or into a twisted cubic and a generator.
Such transitions provide a good geometric insight to understand the systematic classification by Segre symbols as presented in  table \ref{Class.Pencil.P3}.

\begin{definition}
 A quadric is a projective variety defined as the vanishing locus in $\mathbb{P}^n$ of a degree two homogeneous polynomial $Q$ (a quadratic form).
The polynomial $Q$ is given in terms of a
$(n+1)\times (n+1)$ symmetric matrix $\hat{Q}$ as  $Q= x^T \hat{Q} x$, where $x^T=[x_0:x_1:\cdots: x_n]$ is the transpose of $x$  (the projective coordinates of $\mathbb{P}^n$). In this notation, we consider $x$ to be a column vector.
\end{definition}

\subsubsection*{Degeneration of conics and quadric surfaces}
The  quadric hypersurface   $Q= x^T \hat{Q} x$ is non-singular iff the matrix $\hat{Q}$ is non-singular.
The determinant and the minors of the defining  matrix $\hat{Q}$ can be used to describe the degenerations of the quadric $Q$. For example, a quadric in $\mathbb{P}^2$ is usually referred to as a conic. It degenerates into a pair of lines if the determinant of its defining matrix is zero. Furthermore, these two lines coincide if all the first minors\footnote{ The first minors are the determinants of sub-matrices of $\hat{Q}$ obtained by removing one  row and one column. See definition \ref{minors} on page \pageref{minors}. } of the defining matrix vanish.
In the same way, a non-singular quadric surface in $\mathbb{P}^3$ is isomorphic to the Hirzebruch surface $\mathbb{F}_0=\mathbb{P}^1\times \mathbb{P}^1$.
A quadric surface degenerates into a cone if the determinant of its defining matrix is zero. The quadric surface degenerates into a pair of planes if all the first minors of its defining matrix are zero and the two planes coincide if all the second minors vanish as well.

\subsubsection*{Segre embedding and double ruling }
A smooth quadric surface in $\mathbb{P}^3$ is isomorphic to the Hirzebruch surface $\mathbb{F}_0=\mathbb{P}^1\times\mathbb{P}^1$. It can always be expressed as $$x_1 x_4- x_2 x_3=0,$$ where $[x_1:x_2:x_3:x_4]$ are projective coordinates of $\mathbb{P}^3$.
The isomorphism between a quadric surface and the Hirzebruch surface $\mathbb{F}_0$ is  given explicitly by the {\em Segre embedding}.
Let us denote the projective coordinates of $\mathbb{F}_0=\mathbb{P}^1\times\mathbb{P}^1$ as $[s:t]\times[u:v]$. The Segre embedding  is then
$$\mathbb{F}_0\rightarrow \mathbb{P}^3:[s:t]\times[u:v]\mapsto [x_1:x_2:x_3:x_4]=[s u: s v: t  u : t v].$$
A quadric surface admits two different rulings given by each of the two $\mathbb{P}^1$ factors in $\mathbb{F}_0$.
Generators of one  these rulings is called a  line of the quadric surface.
A generator for the first (resp. second) ruling is given by a linear equation in $[u:v]$ (resp. $[s:t]$) and is parametrized by $[s:t]$ (resp.$[u:v]$).  Two distinct generators in the same ruling do not intersect while two distinct generators in different rulings intersect at a unique point.

\subsubsection*{Picard group and bidegree}
The Picard group of a nonsingular quadric surface is $\mathbb{Z}\oplus \mathbb{Z}$ and each of its two generators corresponds to a fiber of one of its two rulings. These two classes intersect at a point and have zero self-intersection.
It follows that curves lying on a nonsingular quadric surface are classified by their bidegree. A curve of bidegree $(p,q)$  is given by a bi-homogeneous polynomial of  degree $p$  in $[s:t]$ and $q$ in $[u:v]$.
\subsubsection*{Intersection numbers and genus} A curve of bidegree $(p,q)$ meets a generator of the first (resp. second) ruling at  $p$ (resp. $q$) points. A smooth curve of bidegree $(p_1,q_1)$ intersects a smooth curve of bidegree $(p_2, q_2)$ at $p_1 q_2+p_2 q_1$ points.
A smooth curve of bidegree $(p,q)$ has  genus $g=(p-1)(q-1)$. We see immediately, that rational curves (curves of genus 0) are those with $p=1$ or $q=1$.
All the curves of bidegreee $(p,q)$ with $p>2$ or $q>2$) are hyperelliptic (genus 2 or higher) while the curves of bidegree $(2,2)$ are elliptic (genus 1).

\subsubsection*{Special curves.}
Certain curves play a central role in our analysis.
A line of $\mathbb{P}^3$ contained in the quadric surface is a  rational curve of bidegree  $(1,0)$ or $(0,1)$. It is called a  {\em generator} of the quadric surface since it is a fiber of one of the two rulings of the quadric surface.   A rational curve of bidegree $(1,1)$ is a conic.
 A rational curve of  bidegree $(1,2)$ or $(2,1)$ is a {\em space cubic} also called a {\em twisted cubic}.
 A curve of bidegree $(2,2)$ is an elliptic curve.

\begin{table}[htb]\label{table.Curves}
\begin{center}
\begin{tabular}{|c|c|c|c|}
\hline
Irreducible curves in a quadric surface & Bidegree & Genus & Degree \\
\hline
\hline
Generator (line) & $(1,0)$ or $(1,0)$& $0$ & $1$\\
\hline
 Conic & $(1,1)$ & $0$& $2$ \\
\hline
Twisted  cubic (=space cubic)& $(1,2)$ or  $(2,1)$ & $0$& $3$\\
\hline
Elliptic curve & $(2,2)$ & $1$ & $4$ \\
\hline
General rational curve & $(1,p)$ or $(p,1)$ & $0$ & $p+1$\\ \hline
Hyper-elliptic curve & $p>2$ and $q>2$ & $g>1$ & $d>4$\\
\hline
\end{tabular}
\end{center}
\caption{Curves in a smooth quadric surface.}
\end{table}

An elliptic curve in a quadric surface has bidegree $(2,2)$. We want to analyze the possible degeneration of a regular elliptic curve within its homology class.
\subsubsection*{Degeneration  into Kodaira fibers}
If the elliptic curve degenerates without splitting into several components, it can be a quartic nodal curve (Kodaira fiber $I_1$) or a quartic cuspidial curve (Kodaira type $II$). When  the elliptic curve degenerates by splitting into multiple curves, we can use the bidegree to explore the different options. We recall that a curve of bidegree $(1,1)$ is a conic, a curve of bidegree $(2,1)$ or $(1,2)$ is a twisted cubic and a curve of bidegree $(1,0)$ or $(0,1)$ is a generator.
We can see from the relations
 $$(2,2)=(1,0)+(1,2),\quad (2,2)=(1,1)+(1,1),$$  that an elliptic curve can degenerate into a generator and a twisted cubic  or into two conics.  In both cases, the configuration consists of two rational curves meeting at two points (Kodaira fiber $I_2$) or at a double point when the two rational curves are tangent to each other (Kodaira fiber $III$).
Since the twisted cubic could split into a conic and a line and a conic can  split into two lines, the previous system can degenerate further into a triangle composed of a conic and two generators
$$(2,2)=(1,1)+(1,0)+(0,1).$$
This corresponds to a  Kodaira fiber of type $I_3$. 
If the three curves intersect at a common point we have a Kodaira fiber of type $IV$. Since a conic can split into two lines, an elliptic curve can also degenerate into a quadrangle (Kodaira fiber of type $I_4$) composed of four generators, two of each rulings:$$(2,2)=(1,0)+(0,1)+(1,0)+(0,1)$$.

\subsubsection*{Non-Kodaira fibers}
Using the intersection of two quadrics in $\mathbb{P}^3$ to model an elliptic curve, there are also several non-Kodaira fibers that can naturally occur. When the elliptic curve degenerates into two conics, the two conic can coincide giving a double conic. Two generators of the same ruling in the $I_4$ fiber can coincide giving a chain  of rational curves with multiplicity $1-2-1$. Such a configuration can specialize further into a multiple fiber of type $2-2$. Finally if both quadric surfaces degenerates into  cones sharing the same vertex, we can have a fiber composed of four lines meeting at a point. For example, the configuration $I_4$ composed of four lines forming a quadrangle can degenerate into four lines meeting at a point ( a 4-star), which we denote by $I^{*-}_0$. If some of these four lines coincide we can have a bouquet of rational curves with multiplicity $1-1-2$,  $2-2$, $1-3$ or $4$. The bouquet  $2-2$ could also be obtained in a smooth quadric surface, by taking the intersection with a double plane tangent to the quadric surface.

\subsubsection*{Non-equidimensional degeneration}
When an elliptic curve is modelled by the intersection of two quadrics in $\mathbb{P}^3$, the two quadrics could coincide given a double quadric surface as a singular fiber. A further degeneration would give two intersecting double planes. Two double planes could also coincide to give a quadruple plane.  If the two quadrics identically vanish, the fiber is the full $\mathbb{P}^3$.

\section{Segre's classification of pencil of quadrics}\label{Section.Segre}

The classification of pencils of quadrics follows the work of Segre \cite{Segre} and relies on  algebraic methods developed by Weierstrass in his studies of quadratic forms.
We refer to \cite{Quadrics} and chapter XI of \cite{analyticgeometry} for a pedagogic and geometric introduction. A purely algebraic approach is presented in  chapter XIII of the second volume of the classical book by Hodge and Pedoe\cite{HodgePedoe}.

\begin{definition}[Pencil of quadrics]
Given two quadrics $Q_1$ and $Q_2$ in $\mathbb{P}^n$, we can consider the pencil $Q:=\mu Q_1+\lambda Q_2$ where $[\mu:\lambda]\in \mathbb{P}^1$. We will often express it in terms of  $r=\lambda/ \mu$ as  $Q:= Q_1+rQ_2$. In that case, $r=\infty$ just means that $[\mu:\lambda]=[0:1]$.
\end{definition}
The vanishing of the minors of the defining matrix of the pencil $Q$ also have a nice geometric interpretation given by the following lemmas:
\begin{lemma}[Characterization of singularities of the complete intersection of two quadrics] If the intersection of two distinct quadrics $Q_1$ and $Q_2$ has a singular point $p$, then either
\begin{itemize}
\item  the determinant of their pencil is identically zero and both quadrics are singular at $p$
\item or the determinant  of their pencil is identically zero and there is a unique quadric of the pencil that is singular at $p$
\item or the determinant of their pencil is not identically zero and there is a unique quadric $(Q_1+r_0 Q_2)$ that is
 singular at $p$ and $r_0$ is a multiple root of the determinant $det (\hat{Q}_1+r \hat{Q}_2)$.
\end{itemize}
\end{lemma}
In order to describe the singularity of a pencil of quadrics, it is useful to introduce the following definitions.
\begin{definition}[$s$-Cones]
A variety $C$ in $\mathbb{P}^n$ is said to be a {\em cone with vertex $O$} if for any point $o$ in $O$ and any point  $x$ in $C$, the line $ox$ joining the two is contained in $C$.  When a cone admits a vertex which is a linear space of dimension $s$, the cone is said to be an {\em  $s$-cone}.
It is common to abuse the expression by simply calling a $0$-cone a cone.
\end{definition}
\begin{definition}[$s$-minors]\label{minors} A {\em $s$-minor} (or a minor of order $s$) of a matrix $M$ is the determinant of a matrix  obtained by removing $s$ rows and $s$ columns from $M$.
\end{definition}
When the determinant of the pencil is not identically zero, the pencil is said to be {\em non-degenerate}. The singular fibers defined by non-degenerate pencils can be characterized using the following lemma:
\begin{lemma}[$s$-cones in a pencil of quadrics]\label{lemma.cone}
The discriminant of a non-degenerate pencil of quadrics  in $\mathbb{P}^n$ has in general $(n+1)$ distinct roots, each corresponding to a $0$-cone. Assume that a root $r_i$ of the determinant of the pencil is also a root of all its minors up to order  $s_i$  (where $s_i\geq 0$) but does not vanish for at least one minor of order $(s_i+1)$. In such a case, the quadric is a $s_i$-cone  with vertex a  $s_i$-dimensional linear space and directrix a smooth quadric in a linear subspace of dimension $(n-2-s_i)$.
\end{lemma}

Lemma \ref{lemma.cone} is central to the classification of pencils of non-degenerate quadrics in $\mathbb{P}^n$. In order to describe the classification of non-degenerate pencils, we first introduce some notations that organize the essential data contained in the previous lemma.
Given a pencil of quadrics determined by a matrix $\hat{Q}_1+r \hat{Q}_2$, we denote by  $\ell_{ij}$  the minimal  multiplicity of a common root $r_i$ of the determinant and all  the minors of $$\hat{Q}_1+r \hat{Q}_2$$ up to order $j\leq (n+1)$ .
We denote by $s_i\geq 0$ the smallest integer such that $\ell_{i, s_i}=0$.
Following Weierstrass, it is more efficient to introduce the differences $e_{ij}$ of successive $\ell_{ij}$: $$e_{ij}=\ell_{i,j-1}-\ell_{i,j}\geq 0, \quad j=1,\cdots, s_j.$$ We have $m_i=\sum_{j=1}^{j=s_j} e_{ij}$ and
\begin{equation}
\Delta_{r}:=det(\hat{Q}_1+r \hat{Q}_2)=\prod_{i=1}^{p} (r-r_i)^{m_i}=\prod_{i=1}^{p}\prod_{j=1}^{ s_i}(r-r_i)^{e_{ij}}.
\end{equation}
In order to classify pencils of quadrics, following Weierstrass,  it is  useful to introduce the concept of {\em elementary divisors} and {\em characteristic numbers}. 
Segre symbols provide an organizational tool for characteristic numbers of a pencil:
\begin{definition}[Elementary divisors, characteristic numbers and Segre symbols]
The factors $(r-r_i)^{e_{ij}}$ are called {\em elementary divisors} and the exponents $e_{ij}$ are called the {\em characteristic numbers}.
They are efficiently organized using {\em Segre symbols}:
\begin{equation}
\sigma_{\hat{Q}_1+r \hat{Q}_2}=[(e_{11}, \cdots ,e_{1,s_1})\cdots (e_{p,1}, \cdots, e_{p,s_p})],
\end{equation}
where $e_{i,1}\leq e_{i,2}\cdots e_{i,s_p}$.
All the characteristic numbers associated with the same root are enclosed in parentheses while the set of all roots is enclosed in  square brackets. The sum of the characteristic number enclosed in the same parentheses gives the multiplicity of the corresponding root.
\end{definition}

The following theorem  provides the classification of non-degenerate pencils of quadrics using Segre symbols. The proof can be found in \cite{HodgePedoe}.

\begin{theorem}[Characterization of pencils of quadrics by Segre symbols]
Two non-degenerate pencils of quadrics in $\mathbb{P}^n$  are projectively equivalent if and only if they have the same Segre symbol and there is an automorphism of $\mathbb{P}^1$ identifying their roots of identical characteristic numbers.
\end{theorem}
We review in table \ref{Class.Pencil.P3}
the classification of non-degenerate pencils of quadrics in $\mathbb{P}^3$.

In order to analyze the singular fibers of $D_5$ elliptic fibrations, we need to determine when the  determinant of the pencil of quadrics has multiple roots and we need to know what the rank of the matrix associated with the pencil.
The determinant of the pencil is quartic. It admits a double root if and only if its  discriminant $\Delta$ vanishes. It  admits a triple root if and only
if  $F$ and $G$ both vanish.  Finally, the quartic admits a  quadruple root if and only if $(q_4, q_3, q_2, q_1)$ is proportional to $(q_3,q_2,q_1,q_0)$.
The quartic has two double roots if and only if it is the square of a quadric, which implies that $q_4 q_1^2-q_3^2 q_0=2q_3^3+q_4^2 q_1 -3 q_4 q_3 q_2=0$.
These classical results are proven for example in chapter 1 of \cite{Quartic}, which are summarized in table \ref{table.quartic}.

\begin{table}[hbt]
\begin{center}
\begin{tabular}{|c|c|}
\hline
Multiple roots & General conditions\\
\hline
One double root & $\Delta=0$\\
\hline
One triple root & $F=G=0$\\
\hline
Two double roots & $q_4 q_1^2-q_3^2 q_0=2q_3^3+q_4^2 q_1 -3 q_4 q_3 q_2=0$\\
\hline
One quadruple root &     $rank \begin{pmatrix}q_4 & q_3 & q_2 & q_1\\ q_3 & q_2 & q_1 & q_0  \end{pmatrix}=1$\\
\hline
\end{tabular}
\caption{Multiple roots for the quartic  $q_0 + 4 q_1  r+ 6 q_2 r^2+4 q_3  r^3+ q_4 r^4$.\label{table.quartic}}
\end{center}
\end{table}

\section{Analysis of the $D_5$ elliptic fibrations with four sections}

The matrix of the pencil describing  our canonical choice for a  $D_5$ elliptic fibration with four sections is
\begin{equation}
\hat{Q}_1+ r \hat{Q}_2=
\begin{pmatrix}
1-r & 0& 0& -\frac{r}{2} e \\
0 & -1 & 0 & -\frac{r}{2} f \\
0 & 0 & r & - c \\
-\frac{r}{2} e & -\frac{r}{2} f & -c & -a-r d
\end{pmatrix}.\label{def.matrix}
\end{equation}
Computing the discriminant, we get
\begin{equation}\label{qdef}
4\mathrm{det}(\hat{Q}_1+r \hat{Q}_2)=q_0 +4 q_1 r + 6 q_2 r^2 + 4q_3 r^3+ q_4 r^4,
\end{equation}
where
\begin{equation}\label{qi}
q_0=c^2, \quad q_1=\frac{1}{4}(4a-c^2),    \quad q_2= \frac{2}{3}(d-a), \quad q_3=\frac{1}{4}(-4 d -f^2+e^2), \quad q_4= f^2,
\end{equation}
The rank of the matrix $(\hat{Q}_1+r \hat{Q}_2)$ of the pencil will be useful to determine the singular fibers. It is given by the following  lemma which is also summarized in table \ref{table.rank}. The proof of this lemma is by direct computation of the cofactors of different order.
\begin{lemma}[Rank of the pencil of quadrics]
The rank of the matrix in equation \eqref{def.matrix}  is never less than $2$. The  matrix has rank  $3$ for a general point of $\Delta=0$. The rank is  $2$ when  $e= 4( a+d)+c^2-f^2=0$  or  $a=c=0$ or  $f=4 d-e^2=0$ and the corresponding roots are respectively  $r=1$, $r=0$ and $1/r=0$.
\end{lemma}

\begin{table}[hbt]
\begin{center}
\begin{tabular}{|c|c|}
\hline
Rank  of $\hat{Q}_r$ & General conditions\\
\hline
 3 & $\Delta=0$\\
\hline
 2 &
{\begin{tabular}{clc}& $e= 4( a+d)+c^2-f^2=0$ &($r=1$)\\ or & $a=c=0$ & ($r=0$) \\ or & $f=4 d-e^2=0$&($r=\infty$)\end{tabular} }\\
\hline
 0 or  1& never \\
\hline
\end{tabular}
\caption{Rank of the pencil of quadrics.\label{table.rank}}
\end{center}
\end{table}

\begin{remark}[Absence of fibers with components of higher  multiplicity or  dimension] With our choice of canonical model for a $D_5$ elliptic fibration with four sections, we have seen that the  rank is never lower than $2$. It follows from a direct inspection of table \ref{Class.Pencil.P3} that no singular fibers of our model have the Segre symbols $[1(111)]$, $[(211)]$ or $[(1111)]$. They correspond respectively to a double conic, two double lines and a double quadric and they  all have  rank 2 or lower. Moreover, we cannot have type $[(22)]$ since it never happens that all the minors of order 2 have a double root. The symbol $[(22)]$ corresponds to two lines and a double line.  All these fibers ($[1(111)]$, $[(211)]$, $[(1111)]$ and $[(22)]$) are those that are of higher dimension or that have components with multiplicities. We see that our choice of fibration has eliminated them from the spectrum of singular fibers.
\end{remark}

\subsection{Kodaira symbols vs Segre symbols}
Following the previous remark,  we are then left with the 9 symbols  $[112]$, $[11(11)]$, $[13]$, $[1(21)]$, $[22]$, $[(11)(11)]$ , $[4]$, $[(31)]$ and  $[(211)]$. Some of these symbols lead to the same type of Kodaira fibers.  This is because in $\mathbb{P}^3$ a line, a plane conic and a twisted cubic are all rational curves (birationally equivalent to a $\mathbb{P}^1$):
$$
\text{a line}\simeq \text{conic} \simeq \text{twisted cubic}\simeq \mathbb{P}^1.
$$
For example type $[22]$ and type $[11(11)]$ both give Kodaira type $I_2$ (two rational  curves intersecting at two distinct points):
$$
[22],\ [11(11)]\  \Rightarrow I_2
$$
For $[22]$ the two rational curves consist of a twisted cubic and a line and for $[11(11)]$, the two rational curves are both conics. In both cases, when the two rational curves become tangent to each other, we have a fiber of Kodaira type $III$. In terms of Segre symbols, it corresponds to type  $[1(21)]$ and $[4]$, respectively for two tangent conics and the twisted cubic and its tangent line.
$$
[1(21)],\ [4]\   \Rightarrow III.
$$
All the remaining  fibers can be simply understood by further degenerations of the two conics.
If one of the conic degenerates into two lines crossing away from the other conic, we have type Kodaira type $I_3$ (three rational curves intersecting as a triangle). If the two lines  intersect on the conic, we have Kodaira type $IV$ (three rational curves meeting at a point). If both conics degenerate into two lines, we obtain a fiber of Kodaira type $I_4$ (four rational curves intersecting as a quadrangle).

\subsection{Pencils of rank 3}
There are four types of pencils of rank 3. Two of them correspond to irreducible fibers:  the  nodal quartic (Segre symbol $[112]$) and  the cuspidial quartic (Segre symbol $[13]$). The two others are composed of two irreducible components (Segre symbol $[22]$ and $[4]$): a twisted cubic and a projective line. The different between the two reducible fibers of rank 3 is the way the two components intersect: when they intersect at two points, we have the Segre symbol $[22]$  and  when the line is tangent to the twisted cubic we have the Segre symbol $[4]$.

\subsection{Pencils of rank 2}
When the quadric $Q_1+r Q_2$ has rank 2, it means that all the  first minors are zero but at least one second order minor is non-zero.  When the rank is two for $r=r_0$, it is useful to use $Q_1+r_0 Q_2$ as one generator of the quadric and $Q_1$ or $Q_2$ as the other. Since for $r\neq 0$ we have $Q_r= 1/r Q_1+Q_2=Q_1+r Q_2$, therefore we define  $Q_\infty$ as $Q_2$.
Using our choice of elliptic fibration, we have seen that $rank (Q_r)=2$ if and only if  $r=0$ or $r=\infty$ or $r=1$. In these three cases, we can take the defining equation of the elliptic fibration to be $Q_1=Q_2=0$ (for $r=0$ or $r=\infty$) and $Q_1+Q_2=Q_1=0$ for $r=1$.
In all these cases, each of the  two planes will cut  the second quadric along a conic and the two conics will intersect at two points. That is type $I_2$ on Kodaira's list while the Segre symbol is $[11(11)]$. If the two intersecting points coincide, it means that the line defined by the intersecting of the two planes intersects the second conic at double points. This is only possible if it is tangent to the conic. The corresponding Segre symbol is $[1(21)]$. The two conics are then also tangent to each other and we have  Kodaira type $III$. If one of the conic splits into two lines, it means that one of the plane is defined by two directrices passing by the same point of the second quadric.  This corresponds to the Segre symbol $[2(11)]$ and Kodaira type $I_3$ since the second conic and the two directrices form a triangle. When the second conic and the two directrices intersect at the same point, we have the Segre symbol $[(31)]$ and Kodaira type $IV$ (a conic and two lines meeting at a point). When the two quadrics split into planes, we have the Segre symbol $[(11)(11)]$: four screw lines forming a quadrangle. This corresponds to Kodaira type $I_4$. We could consider cases, where some of these lines coincide, but it does not happen in our case. The ultimate case, is the singular case where all of the lines intersect at the same point. This is not degenerate pencil and it is not in Kodaira list. We denote it by $I^*_0$.

\section{Sethi-Vafa-Witten formulas}
In F-theory, the Euler characteristic of an elliptic fibration $\varphi:Y\rightarrow B$ plays an important role in the cancellation of the D3 tadpole in the case of compactification with Calabi-Yau fourfolds \cite{SVW}. It also appears in the condition for the cancellation of anomalies of six dimensional theories resulting from a compactification of  F-theory on a Calabi-Yau threefold \cite{GrassiMorrison2}.  Using the Poincar\'e-Hopf theorem,    the  Euler characteristic of a regular variety can be computed as the  degree of its  total Chern class $\chi(Y)=\int c(Y)$. As such integrals are invariant under proper pushforward of the integrand, we can compute the Euler characteristic $Y$ solely in terms of Chern classes on the base $B$ once a proper pushforward $\varphi_*c(Y)$ is computed, i.e.,
\begin{equation}
\chi(Y)=\int_{Y} c(Y)=\int_{B} \varphi_*c(Y).\nonumber
\end{equation}
An expression of the Euler characteristic of the fibration in terms of topological numbers of the base is commomly referred to as a {\em Sethi-Vafa-Witten formula} in the F-theory literature since these three authors produced the first example of such a formula in their analysis of elliptically fibered Calabi-Yau fourfolds of type $E_8$ \cite{SVW}.  Klemm-Lian-Roan-Yau have obtained general results for Calabi-Yau elliptic fibrations of type $E_n$ ($n=8,7,6$) over a base of arbitrary dimension $d$ \cite{KLRY}. In the case of elliptic fourfolds, we have the following theorem:
\begin{theorem}\label{svw 4fold}[\cite{SVW, KLRY}]
Let $Y\rightarrow B$ an elliptically  fibered Calabi-Yau fourfold respectively of type $E_8$, $E_7$, $E_6$ and $D_5$,  then
$$
\begin{cases}
E_8: \quad \chi(Y)=12 c_1(B) c_2(B) +360 c_1^3(B),\\
E_7: \quad \chi(Y)=12 c_1(B) c_2(B) +144 c_1^3(B),\\
E_6: \quad \chi(Y)=12 c_1(B) c_2(B) +\ 72 c_1^3(B),\\
D_5: \quad \chi(Y)=12 c_1(B) c_2(B) +\ 36 c_1^3(B).
\end{cases}
$$
\end{theorem}
It was later emphazised by Aluffi-Esole \cite{AE1,AE2} that it is much more efficient to consider Sethi-Vafa-Witten formulas for the Euler characteristic as numerical avatars of a much more general relation valid at the level of the total homology Chern classes. In that form, the Sethi-Vafa-Witten formula for $E_n$ ($n=6,7,8$) fibrations takes a particular compact form valid over a base of arbitrary dimension and void of any Calabi-Yau hypothesis \cite{AE2}. From these relations one can easily glean the simple geometric fact that the Euler characteristic of $E_n$ ($n=6,7,8$) elliptic fibrations is but a simple multiple of the Euler characteristic of a hypersurface in the base:
\begin{theorem}[\cite{AE1,AE2}]
Let $\varphi:Y\rightarrow B$ be an elliptic fibration of type $E_n$ ($n=6,7,8$). Such an  elliptic fibration is the zero locus of a section of the line bundle $\mathscr{O}(m)\otimes \pi^{*}\mathscr{L}^m$ on the total space of the (weighted) projective bundle $\pi:\mathbb{P}(\mathscr{E})\rightarrow B$, where $m$ is respectively $(3,4,6)$ for $(E_6, E_7, E_8)$. Then
$$
\varphi_* c(Y)=(10-n)\frac{mL}{1+mL}c(B)=  (10-n)  c(Z_m),
$$
where $Z_m$ is  a smooth hypersurface in the base defined as the zero locus of a section of the line bundle $\mathscr{L}^{m}$.
Moreover, the elliptic fibration is Calabi-Yau, if and only if $c_1(\mathscr{L})=c_1(B)$.
\end{theorem}
We then immediately arrive at the following

\begin{corol}
Let $\varphi:Y\rightarrow B$ be an elliptic fibration of type $E_n$ ($n=6,7,8$) over a base of dimension $d$. Then

$$
\chi(Y)=(10-n)\sum_{k=1}^d (-1)^{k+1} (mL)^{k} c_{d-k}(B),
$$
where $m=3,4,6$ respectively for the $E_6$, $E_7$ and $E_8$ cases.
\end{corol}

\subsection{Sethi-Vafa-Witten for $D_5$ elliptic fibrations}
In this subsection, we obtain a  Sethi-Vafa-Witten formula at the level of the total Chern class for a smooth  $D_5$ elliptic fibration without any Calabi-Yau hypothesis and over a base of arbitrary dimension. We start by computing the pushforward of the total Chern class of the $D_5$ elliptic fibration:
\begin{theorem}\label{SVW formula} Let $\varphi:Y\rightarrow B$ be a $D_5$ elliptic fibration  and $L=c_1(\mathscr{L})$. Then
\begin{equation}
\varphi_{*}c(Y)=\frac{4L(3+5L)}{(1+2L)^2}c(B)=6c(Z_2)-c(Z_{2,2}),
\end{equation}
where $Z_2$ denotes a divisor in the base of class $2L$ and $Z_{2,2}$ denotes a codimension 2 subvariety of the base of class $(2L)^2$.
\end{theorem}
\begin{proof}
Let $H=c_1(\mathscr{O}(1))$ and let $L$ denote both $c_1(\mathscr{L})$ and $\pi^{*}c_1(\mathscr{L})$. Using adjunction along with the exact sequences

\begin{eqnarray*}
0\to T_{\mathbb{P}(\mathscr{E})/B}\to T\mathbb{P}(\mathscr{E})\to \pi^{*}TB\to 0 & \\
0\to \mathscr{O}_{\mathbb{P}(\mathscr{E})}\to \pi^{*}\mathscr{E} \otimes \mathscr{O}(1)\to T_{\mathbb{P}(\mathscr{E})/B}\to 0 & \\
\end{eqnarray*}
we get that
\begin{eqnarray*}
i_{*}c(Y)&=&\frac{(1+H)(1+H+L)^3}{(1+2H+2L)^2}\pi^{*}c(TB) \cap [Y] \\
         &=&\frac{(1+H)(1+H+L)^3(2H+2L)^2}{(1+2H+2L)^2}\pi^{*}c(B),
\end{eqnarray*}
where $i:Y\hookrightarrow \mathbb{P}(\mathscr{E})$ is the inclusion. Thus

\begin{equation}
\varphi_{*}c(Y)=\pi_{*}\left(\frac{(1+H)(1+H+L)^3(2H+2L)^2}{(1+2H+2L)^2}\right)c(B)
\end{equation}
by the projection formula. Then by the pushforward formula of \cite{James} we get that
\begin{equation}
\pi_{*}\left(\frac{(1+H)(1+H+L)^3(2H+2L)^2}{(1+2H+2L)^2}\right)=\frac{4L(3+5L)}{(1+2L)^2}=6\cdot \frac{2L}{1+2L}-\frac{4L^2}{(1+2L)^2}
\end{equation}
from which the theorem follows.
\end{proof}
Exploiting the fact that $\int_Y c(Y)= \int_B \varphi_{*}c(Y)$, we obtain the following
\begin{corol}
The Euler characteristic of a smooth $D_5$ elliptic fibration over a base of dimension $d$ is
\begin{equation}
\chi(Y)=6\chi(Z_2)-\chi(Z_{2,2})= -\sum_{k=1}^{d} (-2)^k (5+k)L^k c_{d-k}(B).
\end{equation}
In particular \begin{equation} \begin{cases} dim\ B=1 ,\quad \chi(Y)=12L, \\ dim\ B=2, \quad \chi(Y)=12Lc_1-28L^2, \\ dim\ B=3 , \quad \chi(Y)=12Lc_2-28L^{2}c_1+64L^3.\end{cases}\end{equation}
\end{corol}
To recover the formula for the Euler characteristic of a $D_5$ Calabi-Yau fourfold as given in Theorem~\ref{svw 4fold} and more generally consider the physical relevance of $D_5$ fibrations, we need the following

\begin{prop} Let $\varphi:Y\rightarrow B$ be a $D_5$ elliptic fibration. Then $Y$ is Calabi-Yau if and only if
$c_1(\mathscr{L})=c_1(B)$.
\end{prop}

\begin{proof}
Again, using adjunction and the exact sequences listed at the outset of the proof of Proposition~\ref{SVW formula}, we get that

\begin{equation}
K_{Y}=\pi^{*}(L-c_1(B)),
\end{equation}
where $L=c_1(\mathscr{L})$. Thus $K_{Y}=0$ if and only if $L=c_1(B)$.
\end{proof}
Using the well known fact \emph{any} Calabi-Yau fourfold $Y$ has arithmetic genus $\chi_{0}(Y)=2=\frac{1}{12}c_1(B)c_2(B)$ (as we will see more explicitly in the next subsection), we obtain the following simplification of the formula for the topological Euler characteristic of a $D_5$ Calabi-Yau fourfold:
\begin{equation}
\chi(Y)=288+36c_1(B)^{3}.
\end{equation}
Thus $\chi(Y)$ only depends on the first Chern class of the anti-canonical bundle of $B$.
Moreover, if $c_1^3(B)$ is odd, $\chi(Y)$ is divisible by $12$ but not by  $24$.

\subsection{Todd class of a $D_5$ elliptic fibration}
In the case $Y$ is a projective variety, the following proposition provides a simple expression for the Todd class of an elliptic fibration of type $D_5$:
\begin{prop}\label{arithmetic genus}Let $\varphi:Y\rightarrow B$ be a $D_5$ elliptic fibration. Then denoting by $Z:=Z_1$ an hypersurface of $B$ such that $\mathscr{O}_B(Z)\cong\mathscr{L}$, then :
\begin{equation}\label{Tod1}
\varphi_{*}Td(Y)=(1-e^{-L})Td(B)=\chi(Z,\mathscr{O}_{Z}).
\end{equation}
\end{prop}
\begin{proof}
As Todd classes are multiplicative with respect to exact sequences just as Chern classes are, we proceed as in the proof of Proposition~\ref{SVW formula}. Similar considerations yield
$$
i_{*}Td(Y)=\frac{H(H+L)^3(1-e^{(-2H-2L)^2})}{(1-e^{-H})(1-e^{-H-L})^{3}}\pi^{*}Td(B)$$
where $i:Y\hookrightarrow \mathbb{P}(\mathscr{E})$ is the inclusion. So again, computing $\varphi_{*}Td(Y)$ amounts to computing
$$
\pi_{*}\left(\frac{H(H+L)^3(1-e^{(-2H-2L)^2})}{(1-e^{-H})(1-e^{-H-L})^{3}}\right)=(1-e^{-L})$$
The first equality of the proposition follows by the pushforward formula of \cite{James}. Keeping in mind that $L=[D]$, the second equality $(1-e^{-L}) Td(B)= \chi(Z,\mathscr{O}_Z)$ follows from the Hirzebruch-Riemann-Roch theorem. More precisely, the structure sheaf sequence $0\rightarrow \mathscr{O}_B(-Z)\rightarrow \mathscr{O}_B\rightarrow \mathscr{O}_Z\rightarrow 0$ gives a locally free resolution of $\mathscr{O}_Y$. The Hirzebruch-Riemann-Roch formula then gives
$$\chi(Z,\mathscr{O}_Z)=\chi(B,\mathscr{O}_B)-\chi(B,\mathscr{O}(-Z))=Td(B)-e^{-[Z]} Td(B)=(1-e^{-L})Td(B).$$
\end{proof}
\begin{remark}
The relation we obtained for the pushforward of the Todd class of a $D_5$ elliptic fibration actually is valid for the $E_n$ ($n=6,7,8$) cases as well, which can be used to check directly that $c_1(B)c_2(B)=24$ for a Calabi-Yau fourfold of type $D_5, E_6, E_7$ and $E_8$. In the appendix \ref{appendix.todd}, we present a more general derivation valid for any flat genus-$g$ curve fibration using the Grothendieck-Riemann-Roch theorem, from which the $D_5$, $E_6$, $E_7$ and $E_8$ cases will be but a corollary. 
 \end{remark}

\subsection{Relations for the Hodge numbers}

Again, using the pushforward formula of \cite{James} and the fact that $c_1(B)c_2(B)=24$ for the base of a Calabi-Yau $E_n$ fourfold, one easily obtains Sethi-Vafa-Witten formulas for the arithmetic genera $\chi_1$ and $\chi_2$ of a Calabi-Yau $E_n$ fourfolds thus giving us linear relations on the non-trivial Hodge numbers of such a fourfold $Y$ by Hirzebruch-Riemann-Roch:
\begin{align}
\begin{cases}
 \chi_1(D_5) = -40 -\  6c_1(B)^3, \quad  \chi_2(D_5) = 204 +\ 24c_1(B)^3, \\
 \chi_1(E_6) = -40 - 12c_1(B)^3,\quad \chi_2(E_6) = 204 +\ 48c_1(B)^3 ,\\
  \chi_1(E_7) = -40 - 24c_1(B)^3,\quad  \chi_2(E_7) = 204 +\  96c_1(B)^3 ,\\
 \chi_1(E_8) = -40 - 60c_1(B)^3,\quad  \chi_2(E_8) = 204 + 240c_1(B)^3,
\end{cases}
\end{align}
where $$\chi_{1}(Y)=h^{1,2}(Y)-h^{1,1}(Y)-h^{1,3}(Y)$$  and $$\chi_{2}(Y)=h^{2,2}(Y)-2h^{1,2}(Y)$$ by Hirzebruch-Riemann-Roch.
We note that since $Y$ is a Calabi-Yau fourfold, we have $h^{1,0}(Y)=h^{2,0}=h^{3,0}=h^{4,0}(Y)-1=0$ and therefore  $h^{1,1}(Y)=b_{2}(Y)$ and $2h^{1,2}(Y)=b_3(Y)$ (where $b_i(Y)$ denotes the ith Betti number). As such, all that is needed to compute the Hodge numbers of such a fibration are its second and third Betti numbers along with the formulas above. So if the second and third Betti numbers can be computed as functions of topological numbers of the base $B$, all non-trivial Hodge numbers would then be dependent solely on the topology of the base.

\section{Weak coupling limits}

The weak coupling limit of F-theory was first introduced by Sen\cite{Sen.Orientifold}, establishing a clear connection between F-theory and type IIB orientifold theories. The procedure involved smoothly deforming the F-theory elliptic fibration until all the fibers become singular. In particular, the fibers consisted only of nodal curves over a dense open subset $U$ of the base $B$, and cuspidal curves on the (closed) complement $B\smallsetminus U$ which was where the type IIB orientifold was to be placed. As nodal curves have $j$-invariant of $\infty$ (which are a special case of \emph{semi-stable} curves in algebro-geometric parlance), and cuspidal curves have an undefined $j$-invariant of $"\frac{0}{0}"$ (which are said to be \emph{unstable} curves), in \cite{AE2} a purely geometric description of a weak coupling limit for an arbitrary elliptic fibration was abstracted from the special case of Sen's limit by choosing a specialization from a \emph{semi-stable} fiber to an \emph{unstable} fiber, and then deforming the elliptic fibration until the stable fiber lies over a dense open subset of the base and the unstable fiber lies over the complement. Thus the more singular fibers an elliptic fibration admits the more possibilities you have to choose from for semi-stable to unstable specializations, and so more potential weak coupling limits to explore (for a detailed description of this program, we again refer the interested reader to\cite{AE2}). $D_5$ fibrations with their rich structure of singular fibers admit a total of ten stable to semi-stable transitions, providing potentially ten avenues in which to pursue weak coupling limits. In particular, in the case of $D_5$ we obtain for the first time a weak coupling limit involving a non-Kodaira fiber, and show that it leads to a type IIB orientifold theory with three (distinct) pairs of brane-image-branes. We also verify the ``\emph{universal tadpole relation}'' corresponding to this type IIB configuration, which is a Chern class identity involving the Chern classes of the elliptic fibration, and Chern classes of divisors in the base corresponding to the orientifold and D-branes. As in \cite{AE1,AE2}, the identity holds without any Calabi-Yau hypothesis and over a base of arbitrary dimension. Furthermore, we show that the type IIB orientifold configuration with three brane-image-brane pairs is the \emph{only} configuration satisfying the universal tadpole relation in the $D_5$ case.

\subsection{Sen's limit}
In the seminal work of Sen\cite{Sen.Orientifold}, the weak coupling limit of F-theory was first introduced as an orientifold limit of a smooth elliptic fibration in Weierstrass form (or an $E_8$ fibration):

\[
Y:y^2=x^3+fxz^2+gz^3
\]
Here, $Y$ sits in a $\mathbb{P}^2$-bundle and $f$ and $g$ are appropriate sections of line bundles over the base. Such a fibration has nodal fibers over a generic point of the discriminant hypersurface $\Delta(:=4f^3+27g^2)=0$ and the nodal curve specializes to a cusp over $f=g=0$. We note that the defining equation of the discriminant locus has  the geometry of a cusp. To obtain a degenerate fibration in which all fibers are singular (and so realize the type IIB scenario), we parameterize the discriminant using the  traditional normalization of the cusp:
$$
h\mapsto (f,g)=(-3h^2,-2h^3),
$$
leading us to define the degenerate fibration
$$
Y_{h}:\quad y^2=x^3-3h^2x-2h^3.
$$
The fibers of $Y_h$ over points $\underline{O}:(h=0)$ are all cuspidal type $II$ fibers (and so unstable), and the fibers over $B\smallsetminus \underline{O}$ are all nodal type $I_1$ fibers (and so semi-stable). To obtain $Y_h$ as a smooth deformation of $Y$, we perturb $f$ and $g$ by adding independent sections multiplied by a (complex) deformation parameter $\epsilon$  to obtain a family of generically smooth fibrations $Y_{h}(\epsilon)$ in such a way that $Y_h$ is the flat limit of $Y_{h}(\epsilon)$ as $\epsilon\rightarrow 0$:
\begin{equation}
\text{Sen's Weak coupling limit}:\quad (I_1\rightarrow II)
\begin{cases}
Y_{h}(\epsilon): y^2z =x^3+ f x z^2 + g z^3\\
f=-3h^2+ \epsilon \eta \\
g=-2 h^3+ \epsilon h \eta+ \epsilon^2\chi.
\end{cases}
\end{equation}
We can associate with this limit a double cover of the base
\begin{equation}
X: \zeta^2-h=0,
\end{equation}
which is branched over the hypersurface $\underline{O}:h=0$.
The discriminant and $j$-invariant take the following form at leading order in $\epsilon$:
\begin{equation}
\Delta\sim \epsilon^2 h^2(\eta^2 +12 h\chi), \quad j\sim \frac{h^4}{\epsilon^2 (\eta^2+12 h\chi )}.
\end{equation}
We then pullback the limiting discriminant $\Delta_{h}:h^2(\eta^2 +12h\chi)=0$ via the projection $\rho:X\rightarrow B$ of the double cover to obtain divisors in $X$ corresponding to the orientifold and the D7-brane:

\begin{equation}
\rho^{*}\Delta_{h}:\zeta^{4}(\eta^2+12\zeta^{2}\chi)=0.
\end{equation}
The orientifold is then located at ${O}:\zeta=0$ and the D7-brane wraps the locus $D:\eta^2+12\zeta^{2}\chi=0$.
Tadpole matching between F-theory and type IIB predicts that
\begin{equation}\label{tadpole.rel}
2\chi(Y)=4\chi(O)+\chi(D),
\end{equation}
where the LHS of equation \eqref{tadpole.rel} corresponds to the F-theory tadpole and the RHS of \eqref{tadpole.rel} corresponds to the type IIB tadpole. As  $D$ has generalized Whitney umbrella singularities (in \cite{AE1} it was descriptively referred to as a \emph{Whitney D7-brane}), its Euler characteristic must be defined in an appropriate manner, as singular varieties admit several generalizations of topological Euler characteristic. Let $\pi:\overline{D}\rightarrow B$ be the normalization of $D$ composed with the projection to $B$ and let $S:\zeta=\eta=\chi=0$ be the pinch locus of $D$ in $X$. Then taking

\begin{equation}\label{pinch}
\chi(D):=\chi(\overline{D})-\chi(S)
\end{equation}
turns out to be a notion of Euler characteristic which satisfies (4.5), as shown in \cite{AE1}. Furthermore, it was also shown in \cite{AE1} that the tadpole relation 
holds at the level of total homology Chern classes (with pinch locus correction as in \ref{pinch}), without any Calabi-Yau hypothesis on $Y$ and over a base
of arbitrary dimension. Indeed, the physical considerations leading to \eqref{pinch} provide a powerful ansatz from a purely geometric perspective, as it is
not at all obvious why such a general Chern class identity should hold. 

\subsection{Geometric generalization}
Weak coupling limits were generalized to other fibrations not in Weierstrass form
such as $E_7$ and $E_6$ fibrations in \cite{AE2}. 
In the weak coupling limit, the discriminant 
factorizes as follows $$\Delta=h^{2+n} \Delta_1 \Delta_2\cdots \Delta_k, \quad J\sim h^{4-n}/(\Delta_1 \Delta_2\cdots \Delta_k), \quad 0 \leq n \leq 4.$$
and $h$ is a section of $\mathscr{L}^{2}$. 
One can also define the double cover the base branched at $h=0$.  This is the variety $\iota:X\rightarrow B$ such that $X:\zeta^2=h$. This is known as the {\em orientifold limit of F-theory}. The orientifold is the invariant locus of $X$ under the involution $\zeta mapsto -\zeta$. This is the divivor $\zeta=0$ in the double cover and it projects to $h=0$ in the base.  If $n=0$ the spectrum is composed of an orientifold and D7-branes wrapping the divisors $\Delta_i$. If $n>0$, we have a bound state of an orientifold and $n$ brane-image-brane pairs wrapping the same divisor $\zeta=0$ in the double cover and there are also branes wrapping the divisors $\Delta_i$. 
The divisors can take some particular shape: 
\begin{enumerate}
\item{\bf An invariant brane} When $\Delta_i$ does not depend on $h$.
\item{\bf A Whitney brane} When $\Delta_i:\eta^2- h \chi=0$, it has the structure of a cone. But in the double cover, its pullback has the structure of a  Whitney umbrella $\rho^*\Delta:\eta^2- \zeta^2 \chi=0$. Such a divisor has double point singularities along the codimension one loci $\eta=\zeta=0$. The singularity enhances to a cuspidial-like singularity at the codimension two loci $\zeta=\eta=\chi=0$. In F-theory, the Euler characteristic of such a singular divisor is defined in \cite{CDE,AE1}. One first normalize the divisor and then takes its stringy Euler characteristic.   
\item{\bf A brane-image-brane } $\Delta_i: \eta^2-h\psi^2=0$. This is a specialization of the  Whitney brane with $\chi=\psi^2$. In such a case, wehen we go to the double cover, we have  a brane-image-brane pair $\varphi^* \Delta_i=D_{i+}+D_{i-}$ with $D_{i\pm}:\eta\pm \zeta \psi=0$. Such a brane-image-brane pair is not in the same homology class as the orientifold. If  $\Delta_i= h-\eta^2$, we obtain in the double cover a  brane-image-brane pair $\rho^*\Delta:D_{i+}+D_{i-}$ with $D_{i\pm }:\eta\pm \xi=0$. Such a brane-image-brane is in the same homology class as the orientifold and coincide with it when $\eta=0$.
\end{enumerate}
Given a weak coupling limit, the physics of D-branes requires that $$8[O]=\sum_k [D_k].$$ 
This condition is naturally satisfied with an elliptic fibration since $\Delta$ is a section of $\mathscr{L}^{12}$.  
Moreover, compairing the contribution of curvature to the D3 tadpole in type IIB and in F-theory, we have the tadpole relation $$2\chi(Y)=4\chi(O)+\sum_k \chi(D_k).$$ 
In the case of $E_n$ ($n=8,7,6$), this physical requirement was shown in \cite{AE1, AE2} to be related to a more general relation true at the level of the total Chern class: $$2\varphi_* c(Y)=4\rho_* c(O)+\sum_k \rho_* c(D_k).$$

In the next section, we will present the first example of a weak coupling limit of a  $D_5$ elliptic
fibration.

\subsection{A $D_5$ limit}

$D5$ elliptic fibrations with four sections have a total of 8 types of singular fibers with a rich  structure of enhancement. It is easy to see (e.g. by glancing at Figure 1) that they naturally lead to $4+3+2+1=10$ different types  transitions of stable to semi-stable fibers:
\begin{align}
& I_1\rightarrow II,  III, IV, I_0^{*-}, \\
& I_2\rightarrow   III, IV, I_0^{*-}, \\
& I_3\rightarrow  IV, I_0^{*-}, \\
& I_4\rightarrow  I_0^{*-}.
\end{align}
As we have expressed the fiber by their Kodaira notation, it is important to keep in mind that some of these Kodaira fibers (namely, $I_2$ and $III$) corresponds to several non-equivalent Segre symbols.

We will present a limit defined by the specializtion $I_{2}\rightarrow III$, which enhances further to an $I_0^{*-}$ fiber, i.e., the non-Kodaira fiber consisting of a bouquet of four $\mathbb{P}^1$s meeting at a point.
A fiber of type $I_2$ can be realized by two conics intersecting at two distinct points  (Segre symbol $[11(11)]$) or by a twisted cubic meeting at secant $[22]$. In the same way, a fiber of type $III$ can  be realized by two conics tangent at a point (Segre symbol $[1(21)]$) or by a twisted cubic and a tangent line (Segre symbol $[4]$). In the case at hand, the fiber $I_2$ is realized by two conics meeting at two points and the fiber $III$ is realized when the two conics become tangent to each other. To be specific, the two conics will be obtained by allowing the quadric $Q_1$ to degenerate into two planes. The intersection of each of these planes with $Q_2$ will give one of the two conics. The intersection of the two planes is a line which generally intersects the second quadric  at two points, which are the points of intersection of the two conics. However, when the line becomes tangent to the second quadric surface, the two conics are tangent to each other and gives a fiber of type $III$ (Segre symbol $[1(21)]$).
The degeneration can be simply expressed by the following conditions
$$
a= \epsilon\chi, \quad c=\epsilon \eta \quad d=h, \quad e=\phi_1+\phi_2, \quad f=\phi_1-\phi_2,
$$
where $\epsilon$ is the deformation parameter. We obtain the following family of fibrations:
\begin{align}\label{D5lim}
Y_{h}(\epsilon):
\begin{cases}
x^2-y^2-z\epsilon(\chi z+\eta w)=0 \\
w^2-x^2-z\Big[hz+(\phi_1+\phi_2)x+(\phi_1-\phi_2)y\Big]=0.
\end{cases}
\end{align}
In the flat limit of $Y_{h}(\epsilon)$ as $\epsilon\rightarrow 0$, we obtain a degenerate fibration $Y_h$, whose fibers over $B\smallsetminus (h=0)$ are of type $I_2$ (realized by the Segre symbol $[11(11)]$: two conics meeting transversally at two points), the fibers above $\underline{O}:(h=0)$ are generically of type $III$ (realized by the Segre symbol $[1(21)]$: two conics tangent at a point), and the fiber enhances further (inside $\underline{O}$) to an $I_0^{*-}$ fiber (i.e., the non-Kodaira fiber consisting of a "bouquet" of four $\mathbb{P}^1$s meeting at a point) when $\phi_1=\phi_2=0$, satisfying necessary conditions for a weak coupling limit as established in \cite{AE2}.
The discriminant and $j$-invariant then take the following form at leading order in $\epsilon$:
\begin{equation}
\Delta\sim \epsilon^2h^2(h-\phi_1^2)(h-\phi_2^2)(h\eta^2-\chi^2), \quad j\sim \frac{h^4}{\epsilon^2 (h-\phi_1^2)(h-\phi_2^2)(h\eta^2-\chi^2)}.
\end{equation}
We see from this expression that as $\epsilon\rightarrow 0$, the $j$-invariant will diverge to infinity, which ensures that the generic fibers will be semi-stable.

This tells us that $h$ is the location of an orientifold in the base and $(h-\phi_1^2)$, $(h-\phi_2^2)$ and $(h\eta^2-\chi^2)$ will lead to brane-image-brane pairs in a double cover of the base.
We then consider the double cover of the base $\rho:X\rightarrow B$, where $X$ is a hypersurface in the total space of $\mathscr{L}$ given by
$$
X:\zeta^2=h,
$$
where $\zeta$ is a section of $\mathscr{L}$.
The fact that $j\sim h^4$ tells us we have a pure orientifold residing at ${O}:(\zeta=0)$. To locate the varieties upon which the D7-branes wrap, we pullback the limiting discriminant $\Delta_{h}:(h^2(h-\phi_1^2)(h-\phi_2^2)(h\eta^2-\chi^2)=0)$ via $\rho$ to obtain the location of the D7-branes:

$$
\rho^{*}\Delta_{h}:\zeta^{4}(\zeta+\phi_1)(\zeta-\phi_1)(\zeta+\phi_2)(\zeta-\phi_2)(\zeta\eta+\chi)(\zeta\eta-\chi)=0.
$$

We then see that we have three pairs of brane-image-branes intersecting the orientifold ${O}:(\zeta=0)$:

$$
D_{1\pm}:\zeta\pm \phi_1=0, \quad D_{2\pm}:\zeta\pm \phi_2=0, \quad D_{3\pm}:\zeta\eta\pm \chi=0.
$$
We note that $D_{1\pm}$ and $D_{2\pm}$ are in the same homological class as the orientifold $O$ while $D_{3\pm}$ are in the class $2[O]$.
Tadpole matching between F-theory and type IIB predicts the following relation
\begin{equation}
2\chi(Y)=4\chi({O})+4\chi(D)+2\chi(D_3)=8\chi({O})+2\chi(D_3),
\end{equation}
where $D$ and $D_3$ are divisors in $X$ of class $[D]=[D_{1\pm}]=[D_{2\pm}]=L$ and $[D_3]=[D_{3\pm}]=2L$. Not only does relation (4.12) indeed hold, we show in the next subsection that relation (4.12) can be obtained by integrating both sides of the following Chern class identity:

\begin{equation}
\varphi_{*}c(Y)=\rho_{*}(4c({O})+c(D_3)).
\end{equation}

\subsection{Universal tadpole relations}\label{tadpole}
As the Chern class identity (4.13) holds without any Calabi-Yau hypothesis on our $D_5$ elliptic fibration $\varphi:Y\rightarrow B$ or any restrictions on the dimension of $B$, in \cite{AE2}, such an identity was coined a ``\emph{universal tadpole relation}". We classify such universal tadpole relations corresponding to configurations of smooth branes arising from the weak coupling limit of a $D_5$ model and find that there is only one such relation, namely (4.13), corresponding to an orientifold and three brane-image-brane pairs. Intersetingly, in \cite{AE2}, it was shown that $E_7$ fibrations admit a unique universal tadpole relation corresponding to an orientifold and \emph{one} brane-image-brane pair and $E_6$ fibrations admit a unique universal tadpole relation correponding to an orientifold and \emph{two} brane-image-brane pairs. The fact that $D_5$ fibrations seem to stand next in line to the $E_7$ and $E_6$ cases respectively as they admit a unique universal tadpole relation corresponding to an orientifold and \emph{three} brane-image-branes is compelling, as $E_7$, $E_6$ and $D_5$ fibrations admit $2=1+1$, $3=2+1$ and $4=3+1$ sections respectively.

A universal tadpole relation for an elliptic fibration is generically of the form:

\begin{equation}
2\varphi_{*}c(Y)=\rho_{*}(\displaystyle\sum_{i}c(D_{i})),
\end{equation}
where the $D_i$s are divisors of class $a_{i}L$ in $X$ corresponding to orientifolds and/or D-branes, and $L=\rho^{*}c_1(\mathscr{L})$. As the (pullback) of the discriminant locus is of class $12L$, we necessarily have $\displaystyle\sum a_i=12$. Now a general divisor $D$ of class $aL$ $(a\in \mathbb{Z})$ has Chern class

\begin{equation}
c(D)=\frac{aL}{1+aL}c(X)=\frac{aL}{1+aL}\left(\frac{1+L}{1+2L}\rho^{*}c(TB)\cap [X]\right),
\end{equation}
thus

\begin{equation}
\rho_{*}c(D)=\frac{aL(1+L)}{(1+aL)(1+2L)}c(TB)\cap 2[B]=\frac{2aL(1+L)}{(1+aL)(1+2L)}c(B)
\end{equation}
by the projection formula. Since we know $\varphi_{*}c(Y)$ by Proposition~\ref{SVW formula}, a universal tadpole relation (after canceling factors of $c(B)$) for a $D_5$ model then takes the following form:

\begin{equation}
2\cdot \frac{4L(3+5L)}{(1+2L)^2}=\frac{1+L}{1+2L}\displaystyle\sum_{i}\frac{2a_{i}L}{1+a_{i}L}.
\end{equation}
To classify all such relations (if they exist), we retrieve the 77 partitions of the number 12 and simply plug them into (4.17) and hope for the best. It turns out that only one partition does the job, namely $1+1+1+1+1+1+1+1+2+2$, which corresponds precisely to the universal tadpole relation arising from the weak coupling limit we found in the previous subsection ($[{O}]=[D]=L$ and $[D_3]=2L$):

\begin{equation}
2\varphi_{*}c(Y)=\rho_{*}(4c({O})+4c(D)+2c(D_3)).
\end{equation}
Integrating both sides of (4.18) yields the numerical relation (4.12) predicted by tadpole mathcing between F-theory and type IIB. We record our findings in the following

\begin{prop}
Let $\varphi:Y\rightarrow B$ be a $D_5$ elliptic fibration. Then $Y$ admits a unique universal tadpole relation corresponding to the Chern class identity (4.18). Futhermore, the universal tadpole relation is realized via the specialization $I_2\rightarrow III$.
\end{prop}
\begin{remark}
With the exception of the four transitions $I_1\rightarrow (II, III, IV, I^{*-}_0)$, all other transitions are specializations of $I_2\rightarrow III$. So we can expect to find other weak coupling limits satisfying the tadpole condition as well. We present some examples for each case.
\end{remark}

\subsection{A Weak coupling limit with a non-Kodaira fiber}
By specializing the limit $I_2\rightarrow III$, we can define a configuration corresponding to the transition $I_2\rightarrow I_0^{*-}$.
The specialization is $\phi_1=\phi_2=0$, and gives
$$
I_2\rightarrow I_0^{*-}:\Delta\propto h^4 \epsilon^2(h \eta^2-\chi^2), \quad j\sim \frac{h^2}{\epsilon^2(h \eta^2-\chi^2)}.
$$
We see that two of the brane-image-brane pairs that we had in the case $I_2\rightarrow III$ are now wrapping the same divisor as the  orientifold. For that specialization, in the weak coupling limit $\epsilon \rightarrow 0$, the generic fiber is of type $I_2$ and it specializes to a fiber of type $I^{*-}_0$ over the orientifold. All together we have an orientifold and 2 pairs of brane-image-branes on top of it and an additional pair of brane-image-brane on $\xi \eta\pm \chi=0$.

If we let only one of the two brane-image-brane pairs to coincide with the orientifold (say we specialize to  $\phi_1=0$), we have a transition $I_2\rightarrow IV$.
$$
I_2\rightarrow IV:\Delta\propto h^3 \epsilon^ (h-\phi_1^2)(h \eta^2-\chi^2), \quad j\sim \frac{h^3}{\epsilon^2(h \eta^2-\chi^2)}.
$$
At weak coupling we have a brane-image-brane pair on top of the orientifold ($D_{1\pm}= O: \xi=0$)  and two other brane-image-brane pairs, namely $D_{2\pm}:\xi\pm \phi_2=0$ and $D_{3\pm}:\xi \eta\pm \chi=0$

\subsection{Other limits}
\subsubsection{$I_2\rightarrow III$}
We will have the same discussion if we consider the following limits which are also of the type $I_2\rightarrow III$, but involve different choices of what the rational curves that form the fiber $I_2$ are:
{\small
\begin{align}
\begin{cases}  a=\frac{1}{4} e^2-\frac{h}{4}, \\ c=\frac{1}{2}(\phi_1+\phi_2), \\  d = \frac{1}{4}e^2+\epsilon \chi, \\  e=\frac{1}{2}(\phi_1-\phi_2),\ f=2 \epsilon \eta,\end{cases}\
\begin{cases}  a=\frac{1}{4}(f^2-c^2- 4 d)+\frac{\epsilon}{2} \chi, \\  c  =\frac{1}{2}(\phi_1+\phi_2),\\ d=\frac{1}{4}(c^2+2 f^2-h),\\
e=\epsilon \eta , \  f =\frac{1}{2}(\phi_1-\phi_2)\end{cases}
\end{align}}
Both lead to the same limit as before
$$
\Delta\propto h^2 \epsilon^2(h-\phi_1^2)(h-\phi_2^2)(h \eta^2-\chi^2), \quad j\sim \frac{h^4}{\epsilon^2(h-\phi_1^2)(h-\phi_2^2)(h \eta^2-\chi^2)}.
$$

We can then perform specialization to $I_2\rightarrow IV$ and $I_2\rightarrow I_0^{*-}$.

\subsection{Weak coupling limits, an overall look}

The brane configurations at weak coupling limit satisfies the condition $8[O]=\sum [D_i]$ and  the tadpole matching conditions that compare the contribution of curvature in type IIB and in F-theory:   $2\chi(Y)=4\chi(0)+\sum_i D_i $. When this condition hold, the $G$-flux in F-theory should be accounted completely by the D7-brane fluxes in the weak coupling limit. Fluxes are present in the orientifold limit for example in presence of a brane-image-brane away from the orientifold. A brane-image-brane coincide with the  orientifold only if  they are in the same homology class. It is interesting to look at the results of the weak coupling limits of $D_5$ fibrations in the continuity of the weak coupling limits of  $E_n$ ($n=8,7,6$) fibrations. One will see that pattern. For example, a $E_{n}$ ($n=8,7,6,5$ and $E_5=D_5$) elliptic fibration can describe up to $(n-8)$ brane-image-brane pairs, it has $(9-n)$ sections and admits singular fibers with up to  $(9-n)$ components. In particular, one of them is a fiber of type $I_{9-n}$. Their weak coupling limits that satisfies the tadpole condition are:

\begin{table}[!h]
\begin{center}
\begin{tabular}{|c|c|l|l|}
\hline
Model & Type & Whitney brane & Brane-Image-brane  \\
\hline
$E_8$ & $I_1\rightarrow I_2$& $8[O]$ &      \\
\hline
$E_7$& $I_1\rightarrow II$ & $6[O]$& $\  [O]$  \\
\hline
$E_7$& $I_2\rightarrow III$ & & $4[O]$    \\
\hline
$E_6$ & $I_2\rightarrow III$& & $\ [O]$, $3[O]$ \\
\hline
$E_6$ & $I_2\rightarrow IV$ & & $\ [O]$, $3[O]$ \\
\hline
$D_5$ &$I_2\rightarrow III$ & & $\  [O]$, $\ [O]$, $2[O]$ \\
\hline
$D_5$ & $I_2\rightarrow IV$ & &  $\  [O]$, $\  [O]$, $2[O]$  \\
\hline
$D_5$ & $I_2\rightarrow I^{*-}_0$  & &  $\  [O]$, $\  [O]$, $2[O]$  \\
\hline
 \end{tabular}
\end{center}
\caption{Geometric weak coupling limit and spectrum. In this table $[O]$ is the homology class of the orientifold in the double cover of the base of the elliptic fibration. In the two column, the branes are identified by their homology class. The Whitney branes are always singular with an equation of the type $\eta^2-\zeta^2 \chi=0$. The brane-image-brane are obtained as the factors of a Whitney brane with $\chi=\psi^2$. So they are given by $\eta\pm \eta\psi=0$. When $\psi$ is just a constant, the brane-image-brane pair is constituted of branes in the same homology class as the orientifold. \label{wcl}}
\end{table}

\begin{enumerate}
\item[$E_8$ ($I_1\rightarrow II$):] This is the original example of a weak coupling limit obtained by Sen. The configuration satisfying the tadpole condition corresponds to an orientifold and a Whitney brane $D$ (in the homology class $8[O]$). To satisfy the tadpole condition, the singularities of the Whitney brane have to be taken into account. The appropriate way to do it is to introduce new Euler characteristic $\chi_O(D)$ \cite{CDE, AE1}. The tadpole condition is $2\chi(Y)=4\chi(O)+\chi_o(D)$.

\item[$E_7$ ($I_2\rightarrow III$):] An orientifold and a brane-image-brane pair $D_\pm$ with each branes in the homology class $4[O]$. Each brane is smooth and the configuration satisfies the tadpole condition. This is the only configuration that satisfies the tadpole relation with only smooth branes. 
The tadpole condition is $\chi(Y)=2\chi(O)+\chi(D)$.
\item[$E_7$ ($I_1\rightarrow II$):] This is the another  configuration that satisfies the tadpole condition for a $E_7$ elliptic fibration.  It corresponds to the transition of a  nodal curve to a cusp, just like the original Sen's limit. However, with the $E_7$ fibration, it leads to an orientifold $O$, a brane-image-brane pair $D_{1\pm}$ with each brane in the same homology class as the orientifold ($[D_1]=[O]$) and a Whitney brane $D_2$ in the homology class  $[D_2]=6[O]$. 
The tadpole condition is $2\chi(Y)=4\chi(O)+2\chi(D_1)+\chi_o(D_2)$.
\item[$E_6$ ($I_2\rightarrow III$):] The fiber $I_2$ is constituted by  a conic in $\mathbb{P}^2$ and a secant line. The fiber $III$ is the limit in which the secant line becomes tangent to the conic. At the weak coupling limit, we get  an orientifold and two brane-image-branes pairs, one pair is constitutes of two branes $D_{1\pm}$ in the same homology class as the orientifold and the other pairs involving two branes $D_{2\pm}$ in the homology class $[D_2]=3[O]$. 
The tadpole condition is $\chi(Y)=3\chi(O)+\chi_(D_2)$.
\item[$E_6$ ($I_2\rightarrow IV$):] This is the specialization of the previous configuration when over the orientifold, the fiber $III$ is replaced by a fiber $IV$ obtained by a degeneration of the conic into two lines. Physically, this happens when the brane-image-brane pair $D_{1\pm}$ composed of two branes in the same homology class as the orientifold coincide  with the orientifold.  
The tadpole condition is unchanged $\chi(Y)=3\chi(O)+\chi_(D_2)$.
\item[$D_5$ ($I_2\rightarrow III$)]: This configuration was obtained by using two  conics meeting at two points and becoming tangent to each other in the orientifold limit. In view of the many ways we can define the two conics there are at least 3 different ways to obtain this configuration.  It leads to an orientifold and three brane-image-branes pairs $D_{1\pm}, D_{2\pm}$ and $D_{3\pm}$ with $D_{1\pm}$ and $D_{2\pm}$ constituted of branes in the same homology class as the orientifold and the third pair  is constituted of branes $D_{3\pm}$ in the homology class $[D_3]=4[O]$.  The tadpole relation is $\chi(Y)=4\chi(O)+\chi_(D_3)$. 
\item[$D_5$ ($I_2\rightarrow IV$)]: This is a specialization of the previous configuration when one of the brane-image-brane pairs $D_{1\pm}$ or $D_{2\pm}$ coincides with the orientifold. The fiber $IV$ is constituted by a conic and two lines all meeting at a common point. It can be seen as the limit of the previous case when over the orientifold, one of the conic splits into two lines. The tadpole condition is unchanged: $\chi(Y)=4\chi(O)+\chi_(D_3)$.

\item[$D_5$ ($I_2\rightarrow I^{*-}_0$)]: This is the first example of a weak coupling limit involving a non-Kodaira fiber. It is a specialization of the previous configuration when the two brane-image-branes pairs $D_{1\pm}$ and $D_{2\pm}$ which are in the same cohomology class as the orientifold actually coincide with it. The tadpole condition is unchanged: $\chi(Y)=4\chi(O)+\chi_(D_3)$.

\end{enumerate}

\section{Conclusion}

In this paper, we have studied the structure of elliptic fibrations $\varphi:Y\rightarrow B$ of type $D_5$ with a view toward F-theory. The generic fiber of a $D_5$ elliptic fibration is a smooth quartic space curve of genus one modeled by the complete intersection of two quadrics in $\mathbb{P}^3$. In the canonical model we consider, the elliptic fibration is endowed with a divisor intersecting every fiber at four distinct points. These four points defines  naturally  four (non-intersecting) sections of the elliptic fibration. 

A generic smooth $D_5$ elliptic fibration admits a rich spectrum of singular fibers composed at most of four intersecting rational curves as summarized in figure \ref{figure.enhancement}.  The classification of these singular fibers is a well studied problem of classical algebraic geometry that is more efficiently reformulated in terms of  pencils of quadrics in $\mathbb{P}^3$ and their corresponding Segre symbols as reviewed in  section \ref{Section.Segre} and sumarized in table  \ref{Class.Pencil.P3}. A $D_5$ elliptic fibration admits fibers that are not in the list of Kodaira. We have reviewed them in  figure \ref{NKfig}. These  non-Kodaira fibers are always located over loci in codimension two or higher in the base. In our canonical model, there is only one non-Kodaira fiber, namely the fiber that we call $I^{*-}_0$ composed of four lines meeting at a common point.  We  have also computed several  topological invariants of $D_5$ elliptic fibrations like the   Euler characteristic, their total Chern class and the Todd class  over a base of arbitrary dimension void of any Calabi-Yau hypothesis.

We have also analyzed birational equivalent models of the $D_5$ elliptic fibration leading to $E_6$ elliptic fibrations and a modified version of the  $E_7$ elliptic fibration. While the $E_6$ birational equivalent model has only its usual $I_1, II, I_2, III$ and $IV$ Kodaira singular fibers \cite{AE2}, the $E_7$ birational equivalent model admits on top of its usual $I_1, II, I_2, III$ Kodaira fiber (see \cite{AE2}) an additional fiber whih is not in Kodaira list and which is  composed of a double conic. An  $E_7$ model can always be expressed as a $D_5$ elliptic fibration with one of the two quadric surface being rigid.  In that framework, the  non-Kodaira singular fiber corresponds to a  Segre symbol $[1(111)]$ (see figure \ref{NKfig}). The non-Kodaira fiber $I^{*-}_0$ of  our canonical $D_5$ model is mapped through the birational equivalence to a fiber of type $IV$ of the $E_6$ model and the double conic of the new $E_7$ model. This illustrates how birational equivalent models can have different fiber structure.

The classification of the singular fibers of a $D_5$ elliptic model can be used to  define interesting gauge theories. This will require specializing the model in order to have certain singular fibers with multiple nodes appearing over  codimension-two loci in the base. If the base is at least of dimension two, this will automatically implies the presence of enhancement of singular fibers in codimension two and three. Such enhancemements do not necessary increase the rank of the fiber as can be seen by analyzing figure \ref{figure.enhancement}. In view of the singular fibers, the candidate non-Abelian gauge groups are $\SU(2)$, $\SU(3)$ and $\SU(4)$.

The list of singular fibers can also be used to determine different weak coupling limits for $D_5$ elliptic fibrations.  Indeed, weak coupling limits are characterized by a transition from a semi-stable fiber to an unstable one \cite{AE2}. In the case of our canonical model, such transitions can be seen in figure \ref{figure.enhancement}. Following the point of view started in \cite{AE1,AE2}, we work over a base of arbitrary dimension and without imposing the Calabi-Yau condition. In this regard, one can consider the physics of F-theory as an inspiration to study surprising aspects of the geometry of elliptic fibrations that would be hard to think of otherwise.  It is an impressive fact that conditions that are used to understand the physics of elliptic fibered Calabi-Yau fourfolds and threefolds and the properties of seven branes end up being true for arbitrary basis and without actually requiring the Calabi-Yau condition. The most fascinating example is probably the geometry of the weak coupling limit of $D_5$ elliptic fibrations. 

In the $D_5$ case, we have presented explicit weak coupling limits leading to a type IIB orientifold theory with a $\mathbb{Z}_2$ orientifold and three brane-image-brane pairs, two of which are in the same homological class as the orientifold. We have shown how to construct cases for which a brane-image-brane pair coincides with the orientifold, and in the extreme case where both of the brane-image-brane pairs coincide with the orientifold we obtain a non-Kodaira fiber $I^{*-}_0$ on top of the orientifold.
In every case, we have shown that a universal tadpole relation holds for the defining elliptic fibration over a base of arbitrary dimension without imposing the Calabi-Yau condition. Tadpole conditions in F-theory come from equating the curvature contribution of the D3 branes in type IIB and in F-theory. When tadpole relations are satisfied, the $G$ flux in F-theory corresponds to the flux in the type IIB orientifold theory. In recent works on phenomenological applications of F-theory, models admitting a non-trivial Abelian sector in their gauge group are the center of much attention. Such models are expected to be generated by  brane-image-brane configurations leaving in the same homology class as the orientifold.

There are many interesting aspects of the physics of $D_5$ elliptic fibrations that we have not discussed in this paper and that we hope to address soon.
For example, the specialization to non-trivial Mordell-Weil groups has interesting connections with extra $U(1)$s in the gauge group. As $D_5$ elliptic fibrations admit multiple sections, one can easily model non-Abelian gauge theories with a non-Abelain sector of type $\SU(4)\times \SU(2)$. It would be interesting to study these gauge theories in detail for theories both in four and six space-time dimensions. In the case of a compactification to a six dimensional theory, the cancellation of anomalies in the presence of a non-trivial Mordell-Weil group would be an interesting case to analyze in detail. $D_5$ elliptic fibrations provide simple yet non-trivial models to study such gauge theories.

\section*{Acknowledgements} 
M.E. would like first to thank Anand Patel for his  patient explanations of several algebraic geometry techniques related to the Groethendieck-Riemann-Rock theorem. He is also grateful to Paolo Aluffi, Andres Collinucci and Frederik Denef for many enlightening discussions along the years on the geometry of the weak coupling limit. He would also like to thank  David Morrison and Washington Taylor for interesting discussions on singular fibers in F-theory. He is  very grateful to Dyonisios Anninos, Frederik Denef and  Andrew Strominger   for their friendship and continuous encouragements. M.E. would like to thank the Simons Workshop 2011, Stanford,  Caltech and UCSB string theory groups for their hospitality  during part of this work. 
J.F. would like to thank Paolo Aluffi for continued guidance and support as well as Mark van Hoeij for interesting discussions and alleviation of computational hardships.
\appendix

\section{Pushforward of the Todd class}\label{appendix.todd}
In this appendix, we compute the pushforward of the Todd class of a fibration $$\varphi:Y\rightarrow B$$ of genus $g$ curves via Grothendieck-Riemann-Roch.
Though we have computed this more directly in the case of a $D_5$ elliptic fibration, the power of Grothendieck-Riemann-Roch will enable us to compute 
the pushforward of the Todd class for any genus-$g$ curve fibration (modulo assumptions made below) $\varphi:Y\to B$ from which the case of a $D_5$ 
fibration is but a corollary.
A special role will be played by the relative dualizing sheaf of the fibration $\omega_{Y/B}$.

To invoke Grothendieck-Riemann-Roch (as well as Grothendieck duality), we assume that the fibration $\varphi:Y\to B$ is given by a map that is both proper (i.e. closed varieties map to closed varieties) and flat (which ensures that all fibers are of constant dimension and constant arithmetic genus). The varieties $Y$ and $B$ are assumed to be smooth. We first recall Grothendieck-Riemann-Roch:

 \begin{theorem}[Grothendieck-Riemann-Roch] Let  $\varphi:Y\rightarrow B$ be a proper map between smooth varieties and ${\mathcal F}$ be a coherent sheaf  on $Y$. Then
\begin{align}
\varphi_*\Big( ch({\mathcal F})\   Td(Y)\Big)=ch\Big(\varphi_!\ {\mathcal F} \Big) Td(B),
\end{align}
where $Td(X)$ is Todd class of a variety $X$ and $ch(\mathcal{F})$  is the Chern  character of the sheaf $\mathcal{F}$.
\end{theorem}

We will prove the following 
\begin{theorem}[Todd class of a genus $g$ curve fibration]
Let $\varphi:Y\rightarrow B$ be a proper and flat morphism between smooth projective varieties such that the generic fiber of $\varphi$ is a curve of genus $g$. Then
\begin{align}
\varphi_*Td(Y)=\big(1-ch(\varphi_* \omega_{Y/B}^\vee)\big ) Td(B),
\end{align}
where $\omega_{Y/B}:=\omega_Y\otimes \varphi^*\omega_B^\vee$ is the relative dualizing sheaf of the fibration.
\end{theorem}

\begin{proof}
As we wish to compute $\varphi_*Td(Y)$ for $Y$ fibration of genus-$g$ curves, we take ${\mathcal F}=\mathscr{O}_Y$ since $ch(\mathscr{O}_Y)=1$. Then by GRR we get
\begin{align}
\varphi_*Td(Y)=ch(\varphi_!\  \mathscr{O}_Y) Td(B).
\end{align}
By definition,  $\varphi_! \mathscr{O}_Y=\sum_{i\geq 0} (-1)^i R^i \varphi_* (\mathscr{O}_Y)$, where $R^i \varphi_*$ denotes the higher direct image  functors \footnote{$R^i\varphi_* \mathcal{F}$ is the right derived functor for $\varphi_*$. It is defined as the sheaf associated with the presheaf $U\mapsto H^i (\varphi^{-1}(U), { \mathcal  F}|_{\varphi^{-1} (U)})$.}. Since fibers of $\varphi$ are curves, by the relative dimensional vanishing theorem we get that $R^i \varphi_* \mathscr{O}_Y=0$ for $i>1$. For a flat fibration of genus $g$ curves, $R^1\varphi \mathscr{O}_Y$ is a locally free sheaf (and so coherent) of rank $g$. Moreover, by definition $R^0\varphi_* {\mathscr{O}_Y}:=\varphi_*(\mathscr{O}_Y)\cong \mathscr{O}_B$ (any function on $Y$ is necessarily constant on the fibers as the fibers are projective varieties), thus
\begin{align}
\varphi_*Td(Y)=ch(\mathscr{O}_B-R^1 \varphi_* \mathscr{O}_Y) Td(B).
\end{align}
Since $R^1\varphi_* \mathscr{O}_Y$ is a locally free sheaf of finite rank and $\varphi$ is flat, we can use Grothendick duality to get that
\begin{equation}
R^1\varphi_*\mathscr{O}_Y=\Big[R^0\varphi_*(\mathscr{O}_Y^{\vee}\otimes \omega_{Y|B})\Big]^\vee=\varphi_*(\mathscr{O}_Y^{\vee}\otimes \omega_{Y/B})^\vee=\big[\varphi_* \omega_{Y/B}\big]^\vee,
\end{equation}
where $\omega_{Y/B}$ is the relative dualizing sheaf of the map $\varphi:Y\rightarrow B$ and the second equality follows from the definition of $R^0$. The theorem then follows:
\begin{align}
\varphi_*Td(Y)=ch(\mathscr{O}_B-\varphi_* \omega_{Y/B}^\vee) Td(B)=(1-ch(\varphi_* \omega_{Y/B}^\vee) Td(B).
\end{align}
\end{proof}

When the total space $Y$ is smooth, the relative dualizing sheaf is given by the formula $\omega_{Y/B}=\omega_Y\otimes \big[\varphi^* \omega_B\big]^{\vee}$.
In particular, for $D_5$, $E_6$, $E_7$ and $E_8$ fibrations\footnote{More concretely, in the $E_i$ cases the total space of the fibration is a hypersurface in a projective bundle of the form $\pi:\mathbb{P}[\mathscr{O}\oplus \mathscr{L}^{a_1}\oplus \mathscr{L}^{a_2}]\rightarrow B$, with $Y$ the zero locus of a section of $\mathscr{O}(3)\otimes \pi^*\mathscr{L}^m$. Using the adjunction formula,  we get that $\omega_{Y/B}\cong \mathscr{L}^{m-a_1-a_2}$. $E_8$, $E_7$ and $E_6$ fibrations  correspond to the cases $(a_1,a_2,m)=(2,3,6)$,  $(a_1,a_2,m)=(1,2,4)$ and $(a_1,a_2,m)=(1,1,3)$ repectively. So in all these cases, we have $m-a_1-a_2=1$ and therefore $\omega_{Y/B}\cong \mathscr{L}$ for $E_i$ ($i=8,7,6$) elliptic fibrations. In the $D_5$ case, we have a projective bundle $\mathbb{P}[\mathscr{O}\oplus \mathscr{L}\oplus\mathscr{L}\oplus\mathscr{L}]$ and $Y$ is the complete intersection of two divisors given by sections of $\mathscr{O}(2)\otimes \pi^*\mathscr{L}^2$. It follows that $\omega_{Y/B}$ is also $\mathscr{L}$.} we have $\omega_{Y}=\varphi^{*}(\mathscr{L}\otimes \omega_{B})$ (or equivalently $K_{Y}=\varphi^{*}(c_1(\mathscr{L})-c_1(B))$ so that $\omega_{Y/B}=\varphi^*\mathscr{L}$, thus the pushforward of their respective Todd classes will all be equal. More generally, we can say that for any elliptic fibration $Y$ such that $K_{Y}=\varphi^{*}(c_1(\mathscr{L})-c_1(B))$ we necessarily then have that $\omega_{Y/B}=\varphi^*\mathscr{L}$, giving us that
$R^1 \varphi_*\mathscr{O}_Y\cong \mathscr{L}^\vee$ by arguments given above. Putting things all together we get that 
\begin{align}
\varphi_*Td(Y)=ch(\mathscr{O}_B -\mathscr{L}^\vee) Td(B)=(1-e^{-L})Td(B)
\end{align}

\begin{remark}
It is important to notice that the line bundle $\mathscr{L}$ that appears in the $D_5, E_6, E_7$ and $E_8$ elliptic fibration is closely related to the structure of the elliptic fibration.  If the fibration admits a section, we can consider the birationally equivalent Weierstrass model $zy^2=x^3+F  x z^2+ G z^3$ written in the projective bundle $\mathbb{P}[\mathscr{O}_B\oplus \mathscr{L}^2\oplus \mathscr{L}^3]$ and  $F$ and $G$  are  sections of $\mathscr{L}^4$ and $\mathscr{L}^{6}$ respectively. The discriminant locus of the fibration is then a section of  $\mathscr{L}^{12}$. When there is no torsion class  in the Picard group $Pic(B)$ of the base $B$, this  is enough to define $\mathscr{L}$ uniquely for a given elliptic fibration $\varphi:Y\rightarrow B$ admitting a section. The Picard group is torsion free if and only if $H_1(B, \mathbb{Z})$ is trivial. For example, when $B$ is a Fano threefold, $Pic(B)$ does not admit any  torsion.
\end{remark}
We can interpret this result geometrically by introducing a divisor $Z$ of $B$ such that  $\mathscr{L}=\mathscr{O}_B(Z)$.
Using the exact sequence $0\rightarrow \mathscr{O}_B (-Z)\rightarrow  \mathscr{O}_B\rightarrow \mathscr{O}_Z\rightarrow 0$ along with additivity of the Chern character on exact sequences, we get that $ch(\mathscr{O}_B -\mathscr{L}^\vee)=ch(\mathscr{O}_Z)$. Using Hirzebruch-Riemann-Roch, we get $ch(\mathscr{O}_Z)Td(B)=\chi(Z,\mathscr{O}_Z)$ and therefore
\begin{align}
\varphi_*Td(Y)=\chi(Z, \mathscr{O}_{Z}).
\end{align}

\end{document}